%% file: main.tex
\documentclass[11pt]{article}
\usepackage[margin=1in]{geometry}
\input{preamble}
\usepackage{thmtools}
\usepackage{mathtools}

\title{Deterministic Lower Bounds for $k$-Edge Connectivity in the Distributed Sketching Model} 
\date{}

\author{Peter Robinson\thanks{Peter Robinson
was supported in part by National Science Foundation (NSF) grant CCF-2402836 Collaborative Research: AF: Medium: The Communication Cost of Distributed Computation}\\
\small{School of Computer \& Cyber Sciences}\\
\small{Augusta University}
\and
Ming Ming Tan\thanks{Corresponding author.
Ming Ming Tan was supported in part by National Science Foundation (NSF)
grant CCF-2348346 CRII: AF: The Impact of Knowledge on the Performance of Distributed Algorithms.} \\
\small{School of Computer \& Cyber Sciences}\\
\small{Augusta University}
}

\newcommand{\urdec}{\ensuremath{\mathsf{UR}^{\subset}_{\text{dec}}}}
\newcommand{\ur}{\ensuremath{\mathsf{UR}^{\subset}}}

\newcommand{\overlap}{\mathsf{UniqueOverlap}}

\DeclareMathOperator{\supp}{\mathsf{supp}}

\begin{document}
\maketitle
\thispagestyle{empty}
\pagestyle{empty}
\begin{abstract}
  \normalsize 
  \input{abstract}

\end{abstract}
\newpage %
\setcounter{tocdepth}{2} %
\tableofcontents
\pagestyle{plain}
\setcounter{page}{1}

\clearpage
\input{intro}

\input{prelim}

\input{lb_graph}

\input{separated}

\input{overlap}

\input{overlap_algo}
\input{simulation}

\input{conclusion}

\section*{Acknowledgements}
The authors would like to thank the anonymous reviewer for pointing out the relationship between the $\overlap$ problem and the Leave-One Index problem of \cite{ashvinkumar2023evaluating}.
\appendix
\section*{Appendix}
\input{app_lemS}
\input{det_algo}

\bibliographystyle{alpha}
\bibliography{refs}
\end{document}

%% file: preamble.tex
\usepackage{amsmath,amsbsy}
\usepackage{amsthm} %
\usepackage{microtype} %
\usepackage{lmodern} %
\usepackage{amsfonts}
\usepackage{setspace}
\usepackage{color}
\usepackage{xcolor}
\usepackage{algorithm}
\usepackage[noend]{algpseudocode}
\usepackage{multirow}
\usepackage{paralist}
\usepackage{xspace}
\usepackage{threeparttable}
\usepackage{booktabs}
\usepackage{caption}
\usepackage{subcaption}
\definecolor{MyGreen}{rgb}{0.1333,0.5451,0.1333}
\usepackage[linktocpage=true, 
	pagebackref=true,colorlinks,
  linkcolor=blue, urlcolor=blue,
	citecolor=MyGreen,
	bookmarks,bookmarksopen,bookmarksnumbered]
	{hyperref}

\renewcommand{\Pr}{\operatorname*{\textbf{\textup{Pr}}}}

\DeclareMathOperator*{\EE}{\textbf{\textup{E}}}

\renewcommand{\epsilon}{\varepsilon}
\newcommand{\set}[1]{\{ #1 \}}
\newcommand{\Set}[1]{\left\{ #1 \right\}}

\newcommand{\lb}{\left}
\newcommand{\rb}{\right}
\newcommand{\lt}{\left}
\newcommand{\rt}{\right}
\newcommand{\md}{\middle}

\makeatletter
\newtheorem*{rep@theorem}{\rep@title}
\newcommand{\newreptheorem}[2]{%
\newenvironment{rep#1}[1]{%
\def\rep@title{#2 \ref{##1}}%
\begin{rep@theorem}[restated]}%
{\end{rep@theorem}}}
\makeatother
\newreptheorem{lemma}{Lemma}
\newreptheorem{theorem}{Theorem}
\newreptheorem{corollary}{Corollary}

\newboolean{short}
\newcommand{\onlyShort}[1]{\ifthenelse{\boolean{short}}{#1}{}}
\newcommand{\onlyLong}[1]{\ifthenelse{\boolean{short}}{}{#1}}

\theoremstyle{definition}
\newtheorem{problem}{Open Problem}
\newtheorem{definition}{Definition}

\theoremstyle{plain}
\newtheorem{lemma}{Lemma}
\newtheorem{theorem}{Theorem}
\newtheorem{corollary}{Corollary}
\newtheorem{claim}{Claim} %
\newtheorem{fact}{Fact} %
\newtheorem{open problem}{Open Problem}

\usepackage[most]{tcolorbox}
\tcbuselibrary{theorems}

\tcolorboxenvironment{definition}{lower separated=false,
                colback=white!97!blue,
  colframe=black,                            %
fonttitle=\bfseries,
colback=white,%
coltitle=black,
enhanced,
boxed title style={colframe=black},
}
\tcolorboxenvironment{Corollary}{lower separated=false,
                colback=white!97!blue,
colframe=white,
fonttitle=\bfseries,
colbacktitle=white!60!gray,
coltitle=black,
enhanced,
boxed title style={colframe=black},
}
\tcolorboxenvironment{open problem}{colback=white!95!gray,
colframe=white, 
fonttitle=\bfseries,
colbacktitle=white!60!gray,
coltitle=black,
enhanced,
}

\newcommand{\ann}[1]{%
\text{\footnotesize(#1)}\quad}

%% file: abstract.tex
We study the $k$-edge connectivity problem on undirected graphs in the distributed sketching model, where we have $n$ nodes and a referee.
Each node sends a single message to the referee based on its 1-hop neighborhood in the graph, and the referee must decide whether the graph is $k$-edge connected by taking into account the received messages. 

We present the first lower bound for deciding a graph connectivity problem in this model with a deterministic algorithm. 
Concretely, we show that the worst case message length is $\Omega\lt( k \rt)$ bits for $k$-edge connectivity, for any super-constant $k = O(\sqrt{n})$. 
Previously, only a lower bound of $\Omega\lt( \log^3 n \rt)$ bits was known for ($1$-edge) connectivity, due to Yu (SODA 2021). 
In fact, our result is the first super-polylogarithmic lower bound for a connectivity decision problem in the distributed graph sketching model. 

To obtain our result, we introduce a new lower bound graph construction, as well as a new communication complexity problem that we call $\overlap$.
As this problem does not appear to be amenable to reductions to existing hard problems such as set disjointness or indexing due to correlations between the inputs of the three players, we leverage results from cross-intersecting set families to prove the hardness of $\overlap$ for deterministic algorithms in the 3-party model with simultaneous messages. 

Finally, we obtain the sought lower bound for deciding $k$-edge connectivity via a novel simulation argument that, in contrast to previous works, does not introduce any probability of error and thus works for deterministic algorithms.

%% file: intro.tex
\section{Introduction} \label{sec:intro}

We consider the distributed graph sketching model~\cite{becker2011adding}, where $n$ distributed nodes each observe their list of neighbors in a graph.
Every node sends a single message to a central entity, called referee, who must compute a function of the graph based on the received messages.
In this setting, any graph problem can be solved trivially by instructing each one of the $n$ nodes to simply include their entire neighborhood information in the message to the referee. 
To obtain more efficient algorithms that scale to large graphs, the main goal is to keep the maximum size of these messages (also called \emph{sketches}) as small as possible, preferably polylogarithmic in $n$.
In contrast, the full list of neighbors of a node could be as large as $\Theta\lt( n \rt)$ bits, under the standard assumption that the node IDs are a permutation of $\set{1,\ldots,n}$. 
Consequently, nodes can afford to convey only incomplete information (hence the name ``sketch'') of their local neighbors to the referee, which may paint a somewhat ambiguous picture of the actual graph. 
While it may be tempting to conclude that small sketches do not allow us to solve interesting graph problems in this model, the surprising breakthrough of Ahn, Guha, and McGregor~\cite{AGM-soda12} showed that several fundamental graph problems, such as connectivity and spanning trees, can indeed be solved with sketches of only $O\lt( \log^3n \rt)$ bits, if nodes have access to shared randomness and one is willing to accept a polynomially small probability of error. 
Their technique (called ``AGM sketches'') paved the way for sketch-based graph algorithms for many other graph problems, including vertex connectivity~\cite{guha2015vertex} and approximate graph cuts~\cite{ahn2012graph}; see the survey of \cite{mcgregor2014graph} for a list of additional graph problems.\footnote{Strictly speaking, \cite{AGM-soda12} showed these results for fully dynamic graph streams in the semi-streaming model, but it is straightforward to adapt them to the distributed graph sketching model, as observed in \cite{BMRT-sirocco14}.}

\subsection{Randomized Lower Bounds for Connectivity Problems} \label{sec:rand_lbs}
The question whether AGM sketches are optimal for connectivity problems remained an open problem for several years, until the seminal work by Nelson and Yu~\cite{NY19-soda} showed that $\Omega\lt( \log^3n \rt)$ bits are indeed required for computing spanning forests in the distributed graph sketching model.
The main idea of their lower bound is to use a reduction to a variant of the universal relation problem~\cite{karchmer1995super} in 2-party communication complexity, denoted by $\ur$, where Alice starts with a subset $S$ of some universe $U$, and Bob gets a proper subset $T \subset S$. 
Alice sends a single message to Bob, who must output some element in $S\setminus T$. In \cite{kapralov2017optimal}, Kapralov, Nelson, Pachoki, Wang, and Woodruf showed that $\Omega\lt( \log^{3} n \rt)$ bits are required for solving $\ur$ with high probability. 
To obtain a reduction to $\ur$ in the distributed sketching model, \cite{NY19-soda} define a tripartite graph on vertex sets $V^l$, $V^m$, and a much smaller set $V^r$. 
Each node $v \in V^m$ has a set of unique neighbors in $V^l$ (not shared with any other node in $V^m$) as well as some neighbors in $V^r$.
Since every node in $V^l$ is connected to at most one node $v \in V^m$, any spanning forest must include some edge between $v$ and a node in $V^r$.
The neighborhood of a randomly chosen $v^* \in V^m$ is determined by the input of $\ur$ such that Alice's set $S$ corresponds to all of $v^*$'s neighbors, whereby Bob's set $T$ specifies the subset of neighbors in $V^l$. 
They show that the players can jointly simulate the nodes in $V^l$ and $V^m$ (but not $V^r$!), whereby Bob also simulates the referee.
From the resulting spanning forest, Bob can extract an edge between $v^*$ and $V^r$, which corresponds to an element in $S\setminus T$, and thus solves $\ur$.
Notably absent in their simulation are the nodes in $V^r$, which neither Alice or Bob can simulate, due to their lack of knowledge of set $S \setminus T$.
In fact, trying to faithfully simulate \emph{every} node executing a distributed sketching algorithm in a 2-party model poses a major technical challenge since every edge is shared between its neighbors. 
However, \cite{NY19-soda} show that choosing $V^r$ to be polynomially smaller than $V^m$ suffices to overcome this obstacle:
Since each $w \in V^r$ will be a neighbor of many $v \in V^m$ and $w$ does not know which of them is the important node $v^*$ (as all of these neighborhoods are sampled from the same hard input distribution for $\ur$), it follows that the amount of information that the referee can learn about the neighbors of $v^*$ from the sketches of the nodes in $V^r$ is only (roughly) $\tilde O(\frac{|V^r|}{|V^m|}) = o(1)$, and thus negligible.
Hence, Alice and Bob can omit these messages in their simulation altogether, while only causing a small increase in the overall error probability, due to Pinsker's inequality~\cite{pinsker1964information}. 
We remark that the actual lower bound construction of \cite{NY19-soda} uses multiple copies of this basic building block.

More recently, Robinson~\cite{robinson2023distributed} extended the approach of \cite{NY19-soda} to constructing a $k$-edge connected spanning subgraph ($k$-ECSS), with the main difference being a modified graph construction that enables a reduction to the multi-output variant of universal relation, called $\ur_k$, which has a lower bound of $\Omega\lt( k\log^{2} n \rt)$ bits~\cite{kapralov2017optimal} and requires Bob to find $k$ elements of $S \setminus T$.
The simulation argument in \cite{robinson2023distributed} does not directly use Pinsker's inequality as in \cite{NY19-soda}, and instead obtains a slightly smaller, but nevertheless nonzero, probability of error, by devising a mechanism of reconstructing the output when omitting the sketches of the nodes in $V^r$.
At the risk of stating the obvious, we point out that this result does not have any implications for the problem of \emph{deciding} whether a given graph is $k$-edge connected, which is the focus of the current paper. 

In~\cite{Y-soda21}, Yu proved that the decision problem of graph connectivity is subject to the same lower bound on the message length as computing a spanning forest, which fully resolved the complexity of connectivity in the distributed sketching model. 
For this purpose, \cite{Y-soda21} introduced and proved the hardness of a new decision version of universal relation, called $\urdec$.
Just as for the $\ur$ problem, Alice again gets a set $S \subseteq U$.
However, Bob not only gets $T \subset S$ but, in addition, he also knows a partition $(P_1,P_2)$ of $U \setminus T$. The inputs adhere to the promise that either $S\setminus T \subseteq P_1$ or $S\setminus T \subseteq P_2$, and, upon receiving Alice's message, Bob needs to decide which is the case. 
In order to embed an instance of $\urdec$ in a graph, \cite{Y-soda21} adapted the construction of \cite{NY19-soda} by partitioning the set $V^r$ into $V^r_1$ and $V^r_2$, and by sampling the neighborhood of each node in $V^m$ according to the hard distribution of $\urdec$. 
Following the overall simulation approach of \cite{NY19-soda}, Alice's input of $\urdec$ is embedded at some special node $v^* \in V^m$. Accordingly, the partition-promise of $\urdec$ translates to each $v \in V^m$ having all its $V^r$-neighbors in either $V^r_1$ or $V^r_2$.

We point out that there does not appear to be a natural extension of the connectivity lower bound of \cite{Y-soda21} to $k$-edge connectivity. Even though $\urdec$ can be generalized to $k>1$, proving a lower bound that scales linearly with $k$ for this generalization appears to be far from straightforward. Interestingly, the situation is reversed for the standard ``search'' version of $\ur$, where the proof of the lower bound for $\ur_k$ is technically less involved than the corresponding result for $\ur$, see \cite{kapralov2017optimal}. 
A second obstacle is that, in contrast to the problem of finding a $k$-edge connected spanning subgraph~\cite{robinson2023distributed}, designing a suitable lower bound graph for deciding $k$-edge connectivity is itself nontrivial.
This is because choosing the neighborhood of every node in $V^m$ according to the hard distribution of the generalized variant of $\urdec$ and embedding it in the constructions of \cite{NY19-soda,Y-soda21} may result in a graph that is very likely to be $k$-edge connected, independently of the neighborhood of the special node $v^*$.

\subsection{Deterministic Lower Bounds for Graph Connectivity Problems?} \label{sec:determ}
Interestingly, none of the above mentioned lower bound results provide a straightforward way towards obtaining stronger bounds for deterministic graph sketching algorithms (that never fail): 
Since any deterministic one-way algorithm for $\ur$, $\ur_k$, or $\urdec$ needs to send Alice's entire input to Bob,\footnote{While there is no published deterministic lower bound for $\urdec$, the same communication complexity bound as for the search version $\ur$ follows from standard fooling set and counting arguments for one-way communication; see, e.g., \cite{KN-book97}.} it may appear that, at first glance, we can simply ``plug in'' these results for deterministic protocols and directly obtain stronger lower bounds in the graph sketching model.
This intuition, however, turns out to be misleading, as a major technical challenge is that the simulations used in \cite{NY19-soda,Y-soda21,robinson2023distributed} are themselves randomized and, more importantly, are necessarily ``lossy'' in the sense that they introduce a nonzero probability of error, due to forgoing the simulation of some nodes---in particular, all nodes in $V^r$.
This poses a major technical obstacle towards obtaining lower bounds for algorithms that fail with exponentially small probability, including deterministic ones.\footnote{We discuss prior work on lower bounds for other graph problems in the sketching model in Section~\ref{sec:additional}.}

Obtaining lower bounds on the \emph{deterministic} complexity of distributed graph sketches for connectivity is listed under ``hard open problems'' as a ``longstanding open question'' in \cite{assadi2022lower} (p.\ 102).  Yet, there has been no progress on this class of problems to date, which is indicative of the difficulty of finding lower bounds for deterministic connectivity problems. Thus, besides the actual lower bound for $k$-connectivity, a major contribution of our work is to show that 0-error simulations are feasible in this model, which may pave the way for other deterministic lower bounds in graph sketching.

\subsection{Our Contributions}

In this work, we take the first step toward determining lower bounds for deterministic connectivity algorithms in the graph sketching model. 
In more detail, we focus on the problem of deciding $k$-edge connectivity, and we show that any deterministic sketching algorithm has a worst case sketch length that is almost linear (in $k$):
\newcommand{\thmMain}{
Every deterministic algorithm that decides $k$-edge connectivity on $n$-node graphs in the distributed sketching model has a worst case message length of $\Omega( k )$ bits, for any super-constant $k = k(n) \le \gamma\sqrt{n}$, where $\gamma>0$ is a suitable constant. %
}
\begin{theorem} \label{thm:main}
\thmMain
\end{theorem}
The proof of Theorem~\ref{thm:main} requires us to overcome several technical challenges:
\begin{itemize} 
\item To obtain a graph that is hard for deciding $k$-edge connectivity in the distributed sketching model, we need to find a suitable graph construction.
We cannot directly extend the constructions used for deciding graph connectivity in \cite{Y-soda21}, as this crucially rests on the assumption that the neighborhood of every node in $V^{m}$ (see Sec.~\ref{sec:rand_lbs}) must be chosen from the \emph{same} distribution. 
In fact, the straightforward generalization of this idea to $k$-edge connectivity would yield a graph that is very likely to be $k$-edge connected, for large values of $k$.
In this work, we propose a new construction that does not preserve the uniformity property of \cite{Y-soda21} regarding the input distribution. 
Instead, there is a special node $v_\sigma$, with a somewhat larger degree than its peers, which we connect in a way such that the $k$-edge connectivity of the graph entirely rests on the choice of neighborhood for this node. Nevertheless, we ensure that its neighbors have no easy way of distinguishing $v_\sigma$ from the vast majority of unimportant nodes. 
This approach introduces additional technical difficulties that we address in our simulation. 

\item As elaborated in more detail in Section~\ref{sec:determ}, we need to depart from the well-trodden path of using the ``partial simulation'' argument pioneered in \cite{NY19-soda} (and also used in the work of \cite{Y-soda21} and \cite{robinson2023distributed}), where some nodes are omitted from being simulated, as this would result in a nonzero error algorithm. 
To avoid this pitfall, we design a completely different simulation that is suitable for deterministic algorithms.  
We point out that the main technical difficulty of the distributed sketching model, namely the ``sharing'' of edges between nodes, also surfaces in our setting and prevents us from devising a simulation, where every node $v$'s neighborhood is part of the input of one of the simulating players.
Instead, we exploit some structural properties of the algorithm at hand that must hold for any deterministic protocol.
In more detail, the main idea behind our approach is that it suffices to compute a sketch for $v$ that looks ``compatible'' with our (limited) knowledge of $v$'s neighborhood, and we prove that this indeed ensures an error free simulation.

\item Our simulation works in a simultaneous 3-party model of communication complexity, where two players, Alice and Bob, each send a single message to Charlie who computes the output. 
For this setting, we introduce the $\overlap$ problem, which abstracts the core difficulty present in our lower bound graph construction: 
That is, nodes are unaware which of their incident edges point to the crucial node $v_\sigma$, even though $v_\sigma$ itself is aware of its special role.
We capture this property in the communication complexity setting by equipping Alice and Bob each with an input vector that has a single common coordinate whose XOR is $1$.
While Charlie knows the location of this coordinate, he does not know the values of the respective bits of Alice and Bob, which are needed for the correct output.
The hardness of $\overlap$ does not appear to follow from reductions to standard problems such as set disjointness or variants of indexing (see Sec.~\ref{sec:overlap_rel} for a more detailed discussion), due to the correlation between the inputs of the three players and since we consider ternary input vectors, whose entries may be either $0$, $1$, or $\perp$.
A further complication is that we need to assume Charlie knows the entire support (i.e., the indices of all non-$\perp$ entries) of Alice's and Bob's inputs.
Instead, we devise a purely combinatorial argument based on properties of cross-intersecting set families to show the following result:
\end{itemize}
  
\newcommand{\thmOverlap}{
For any sufficiently large $m$, there exists an $s=\Theta\lt( m \rt)$ such that the deterministic one-way communication complexity of $\overlap_{m,s}$ in the simultaneous 3-party model is $\Omega\lt( m \rt)$, where each $m$-length input vector has support $s$. 
}
\begin{theorem} \label{thm:overlap}
\thmOverlap
\end{theorem}
 
\subsection{Additional Related Work} \label{sec:additional}

The distributed graph sketching model was first introduced by Becker, Matamala, Nisse, Rapaport, Suchan, and
Todinca in~\cite{becker2011adding}, where they proved the hardness of several fundamental graph problems for deterministic sketching algorithms:
In particular, they showed that problems including deciding whether the diameter is at most $3$, and whether the graph contains a certain subgraph such as a triangle or a square, require sketches of large (i.e., super-polylogarithmic) size.
Their approach works by showing that the existence of an efficient protocol for any of these problems can be leveraged for reconstructing graphs with certain properties using similar-sized sketches.
However, as pointed out in \cite{becker2011adding}, their approach does not extend to graph connectivity problems.

The model was further explored by Becker, Montealegre, Rapaport, and Todinca in \cite{BMRT-sirocco14}, where they show separations between the power of deterministic, private randomness and public randomness algorithms.
Apart from the aforementioned works of \cite{NY19-soda,Y-soda21,robinson2023distributed} on connectivity-related problems, Assadi, Kol, and Oshman~\cite{AKO-podc20} showed a lower bound of $\Omega\lb(n^{1/2-\epsilon}\rb)$ bits on the sketch size for computing a maximal independent set.
For the distributed sketching model with \emph{private} randomness, Holm, King, Thorup, Zamir, and Zwick~\cite{HKTZZ-focs19} showed that computing a spanning tree is possible with sketches of $O(\sqrt{n}\log n)$ bits.

We point out that the distributed graph sketching model is known to be equivalent to a single-round variant of the broadcast congested clique, as observed in several earlier works~\cite{jurdzinski2018communication,AKO-podc20}, where each one of $n$ nodes can broadcast a single message per round that becomes visible to all other nodes at the end of the current round.
In that context, the worst case sketch size corresponds to the available bandwidth, i.e., size of the broadcast message.
Several works have studied tradeoffs between the number of rounds and the message size in the (multi-round) broadcast congested clique:
Pai and Pemmaraju~\cite{PP20-fsttcs} showed that $\Omega\lt( \frac{\log n}{b} \rt)$ rounds are required if nodes can broadcast messages of size $b$.
The work of Montealegre and Todinca~\cite{MT16-podc} revealed that there is a deterministic $r$-round connectivity algorithm that sends messages of size $O(n^{1/r}\log n)$.
Subsequently, Jurdzinski and Nowicki~\cite{JN-disc17} obtained an improved round complexity of $O(\log n/\log \log n)$ when considering the standard assumption of $O(\log n)$ bits per message.

The related problem of computing a $k$-edge connectivity certificate has also been studied in dynamic graph streams. 
The recent work of Sawettamalya and Yu~\cite{DBLP:journals/corr/abs-2510-16336} give an algorithm that uses $O\lt( n \log^2n \cdot \max\set{k\log n\log k} \rt)$ bits of space and thus nearly matches the lower bound of $\Omega\lt( n \log^2n \cdot \max\set{k,\log n} \rt)$ established in \cite{NY19-soda} and \cite{robinson2023distributed}.

%% file: prelim.tex
\subsection{Notation and Preliminaries} \label{sec:preliminaries}

In the distributed sketching model each node $v$ is equipped with an ID of $O(\log n)$ bits. 
We sometimes abuse notation and use $v$ to refer to either the node itself or its ID, depending on the context.
Thus, when referring to the neighborhood of a node, we also mean the list of IDs associated to the node's neighbors.

We use $[m]$ to denote the set of integers from $1$ to $m$. 
For a set $S$ and some positive integer $m$, we borrow the standard notation from extremal combinatorics  (e.g., see \cite{jukna2001extremal}) and use $\binom{S}{m}$ to denote the set family of all $m$-element subsets of $S$.

 \subsection{Roadmap}
 
We structure the proof of Theorem~\ref{thm:main} as follows:
In Section~\ref{sec:lb_graph}, we introduce our new graph construction and analyze how its properties relate to $k$-connectivity.
Then, we show in Section~\ref{sec:separated} that, for any given deterministic algorithm that sends short sketches, there exists an instance of this lower bound graph, such that many nodes send the same message for certain pairs of inputs. 
After introducing and proving the hardness of the $\overlap$ problem in Section~\ref{sec:overlap}, we show in Section~\ref{sec:simulation} how to solve it by simulating a $k$-edge connectivity algorithm on our lower bound graph.
In Section~\ref{sec:combining}, we combine these results to complete the proof of Theorem~\ref{thm:main}.
\onlyLong{Finally, we discuss some open problems in Section~\ref{sec:conclusion}.}

\onlyShort{\medskip\noindent\textbf{Omitted Proofs.} We include all omitted proofs in the attached full paper.}

%% file: lb_graph.tex
\section{The Lower Bound Graph} \label{sec:lb_graph}

In this section, we define the class of lower bound graphs for proving Theorem~\ref{thm:main}.
Consider parameters $n$ and $k$, where 
\begin{align}
	\omega\lt( 1 \rt) \le k \le \gamma\sqrt{n}, \label{eq:k}
\end{align}
for some suitable constant $\gamma>0$, which, in particular, rules out $k$ being a constant.
We define the class of $n$-node graphs $\mathcal{G}_{k,n}$ on a vertex set $V \cup W \cup \set{u_A,u_B}$, such that 
\begin{align}
|V| &=n - \lt\lfloor \sqrt{n} \rt\rfloor- 2 \label{eq:sizeV}\\ 
|W| &=\lt\lfloor \sqrt{n} \rt\rfloor.  \label{eq:sizeW}
\end{align}
The set $W$ is further partitioned into vertex sets $A$ and $B$, each of size at least $k$.
Every vertex has a unique \emph{identifier} (ID), and we use any fixed arbitrary ID assignment from the set $[n]$ for all graphs in $\mathcal{G}_{k,n}$. 
Every graph $G \in \mathcal{G}$ is associated with a \emph{determining index} $\sigma \in [|V|]$.  
As we will see below, this name is motivated by the fact that the edges in the cut $E(v_\sigma,A \cup B)$ determine whether the graph is $k$-connected. 

Next, we define the edges of $G$:
\begin{enumerate}
\item[(E1)] $G[A]$ and $G[B]$ are cliques.
\item[(E2)] For every $w_i \in A$, we add the edge $\set{u_A,w_i}$,  and,  similarly, for every $w_j \in B$, we add $\set{u_B,w_j}$.
\end{enumerate}
We now define the edges in the cut $(V,W)$: 
For every $v_i \in V$ ($i \ne \sigma$), the incident edges are such that exactly one of the following two properties hold:
\begin{enumerate}
\item[(E3)] $|E(v_i,A)| \ge 1$ and $|E(v_i,B)| = 0$, or $|E(v_i,A \cup B)| =0$; and there are $k$ parallel edges between $u_A$ and $v_i$.
\item[(E4)] $|E(v_i,A)| = 0$ and $|E(v_i,B)| \ge 1$, and there are $k$ parallel edges between $u_B$ and $v_i$.
\end{enumerate}
We say that a node $v_i$ is \emph{$A$-restricted} if it satisfies (E3) and call $v_i$ \emph{$B$-restricted} if (E4) holds.
Note that (E3) holds for a node $v_i$ even in the case that there are no edges at all between $v_i$ and $A \cup B=W$.

Node $v_\sigma$, on the other hand, has exactly $2k-1$ edges to nodes in $W$, and 
\begin{enumerate}
\item[(E5)] there are $k$ parallel edges between $v_\sigma$ and $u_A$.
\end{enumerate}
Moreover, one of the following two conditions hold:
\begin{enumerate}
\item[(C0)] $|E(v_\sigma,A)| \ge k$ and $|E(v_\sigma,B)| \le k-1$;  \label{c0}
\item[(C1)] $|E(v_\sigma,A)| \le k-1$ and $|E(v_\sigma,B)| \ge k$.   \label{c1}
\end{enumerate}
We point out that (E3), (E4), and (E5) require the input graph to be a multi-graph. This can be avoided by simply replacing $u_A$ and $u_B$ by cliques of size $k$ if desired, without changing the asymptotic bounds.
Apart from the list of their neighbors, we provide additional power to the algorithm by revealing to each node $v_i$ whether it is $A$-restricted, $B$-restricted, or $v_\sigma$.
Clearly, this can only strengthen the lower bound.
\onlyLong{Figure~\ref{fig:lb}}\onlyShort{Figure~1 in the attached full paper} shows an example of a graph $G \in \mathcal{G}_{k,n}$.

\onlyLong{
\begin{figure}[t]
  \centering
\includegraphics[scale=1.0]{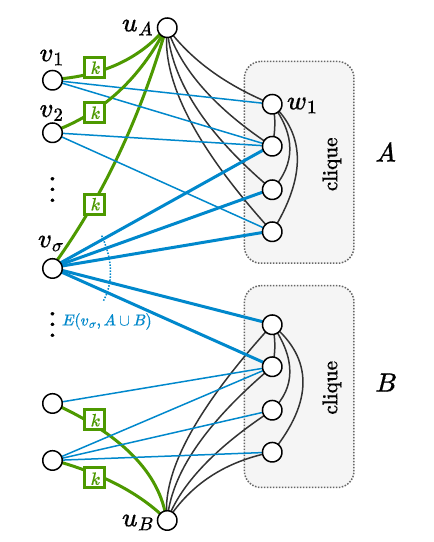}
\caption{\small
An instance of a lower bound graph for $k=3$. Notice that $v_\sigma$ is the only vertex with edges in both $A$ and $B$. In this example, condition (C0) holds, since $v_\sigma$ has $k$ edges to $A$. Thus, the graph is not $k$-edge connected.}
\label{fig:lb}
\end{figure}
}

We now state some properties of this construction.
The next lemma is immediate from the description:
\begin{lemma} \label{lem:distinguish}
Every node $v_i \in V$ is either $v_\sigma$, $A$-restricted, or $B$-restricted.
Moreover, $v_i$ knows this information as part of its initial state.
\end{lemma}

\begin{lemma} \label{lem:lb_graph}
Graph $G \in \mathcal{G}_{k,n}$ is $k$-edge connected if and only if (C1) holds.
\end{lemma}
\begin{proof}
Let $V_A \subseteq V$ be the subset of nodes that are $A$-restricted, i.e., satisfy (E3), and also include $v_\sigma$ in $V_A$. 
Similarly, we define $V_B \subseteq V$ to contain all $B$-restricted nodes, which satisfy (E4).
Every $v_j\in V_A$ has $k$ parallel edges to $u_A$, which itself has an edge to each node in $A$.
Since $A$ is a clique of size at least $k$, it follows that there are at least $k$ edge-disjoint paths between every pair of nodes in $G[A \cup \set{u_A} \cup V_A]$, which implies that $G[A \cup \set{u_A} \cup V_A]$ is $k$-edge connected. A similar argument shows that $B$, $u_B$, and $V_B$ induce a $k$-edge connected subgraph. 
By properties (E1)-(E5), any edge in the cut $(A\cup \set{u_A} \cup V_A,B\cup \set{u_B} \cup V_B)$ must be in $E(v_\sigma, B)$.
Thus, the graph is $k$-edge connected if and only if Case (C1) occurs.
\end{proof}

%% file: separated.tex
\section{Existence of Indistinguishable Separated Pairs} \label{sec:separated}

As our overall goal is to prove Theorem~\ref{thm:main}, we assume towards a contradiction that each node sends at most $L=o\lt( k \rt)$ bits throughout this section.

For a node $v_i \in V$, we call the subset of its neighbors in $W$ its \emph{$W$-neighborhood}.
Here, our main focus will be to show that there is a way to choose the partition $\set{A,B}$ of $W$ that has properties suitable for our simulation in Section~\ref{sec:simulation}.

\begin{definition}[Separated Pair] \label{def:sep_pair}
We say that two $W$-neighborhoods \emph{$S_0$ and $S_1$ form a separated pair for $v_\sigma$ with respect to $A$ and $B$}, when the following are true:
\begin{itemize} 
\item if $v_\sigma$'s $W$-neighborhood corresponds to $S_0$, then condition (C0) holds. In other words, $|S_0 \cap A| \ge k$ and $|S_0 \cap  B| \le k-1.$
\item if $v_\sigma$'s $W$-neighborhood is given by $S_1$, then condition (C1) holds. Thus, $|S_1 \cap A| \le k-1$ and $|S_1 \cap  B| \ge k.$
\end{itemize}
\end{definition}  
Hence, if $v_\sigma$'s $W$-neighborhood is $S_1$, the graph is $k$-edge connected  whereas if it is $S_0$, then the graph is not $k$-edge connected. 

Since node $v_\sigma$ has exactly $2k-1$ neighbors in $W$ according to our lower bound graph description, we know that the set of all $W$-neighborhoods of $v_\sigma$ is of size $\binom{|W|}{2k-1}$.
For technical reasons, however, we need to ensure that any pair of possible $W$-neighborhoods of $v_\sigma$ have an intersection that is at most a small constant fraction of $k$. %
The proof of the next lemma follows via the probabilistic method and similar to Lemma~6 in \cite{kapralov2017optimal}.\footnote{This is Lemma~7 in the full version of their paper on arXiv~\cite{kapralov2017optimalArXiv}.}
We include a full proof in \onlyLong{Appendix~\ref{app:lem:S} for completeness.}\onlyShort{in the appendix of the attached full paper for completeness.}

\newcommand{\lemS}{%
Consider a set family $W$, any small $\epsilon>0$, and any positive integer $d \le \frac{\epsilon |W|}{2e^{2+4/\epsilon}}$.
There exists a set family $\mathcal{S} \subseteq \binom{W}{d}$ such that, for any two distinct $S_1,S_2 \in \mathcal{S}$, we have $|S_1 \cap S_2| \le \frac{\epsilon d}{2}$, and $|\mathcal{S}| \ge  e^{d-1}$. %
}
\begin{lemma} %
\label{lem:S}
\lemS
\end{lemma}

Instantiating Lemma~\ref{lem:S} with $d=2k-1$ and a suitable small constant $\epsilon>0$, shows that
\begin{align}
	|\mathcal{S}| \ge e^{2k-2}.
\label{eq:S}
\end{align}

Given an arbitrary set $D$, we say that $\{\mathcal{B}_1, \ldots, \mathcal{B}_\ell\}$ is a \emph{partition $\mathcal{P}$ of $D$} and each set $\mathcal{B}_i \subseteq D$ is called a \emph{block of $\mathcal{P}$}, if $\bigcup_{i=1}^{\ell} \mathcal{B}_i = D$ and all blocks are pairwise disjoint.  
For a fixed partition $\set{A,B}$ of $W$ and any $S \in \mathcal{S}$, we define the \emph{projections} $\phi_A(S) = S \cap A$ and $\phi_B(S) = S \cap B$, and we also define the resulting image sets $\Phi_{A}=\Set{ \phi_A(S) \ \md|\ S \in \mathcal{S}}$ and $\Phi_{B}=\Set{ \phi_B(S) \ \md|\ S \in \mathcal{S}}$.
Consider the algorithm executed by some node $v_i$.
Let us assume that, for the time being, $\sigma=i$, and that node $v_i$ is equipped with ID $i$.
Together with the node IDs in its neighborhood, this fully determines the message $v_i$ sends to the referee. 
In particular, $v_i$ can send at most $2^{L}$ distinct messages, and hence the given algorithm induces a partition $\mathcal{P}_{i,\sigma}$ of $\mathcal{S}$ into at most $2^{L}$ distinct \emph{blocks}, where each block is a set of $W$-neighborhoods for which $v_i$ sends the same message to Charlie. %
Let $\mathcal{B}_{i,\sigma}$ denote the largest one of these blocks, and observe that  
	$\mathcal{B}_{i,\sigma} \ge \frac{|\mathcal{S}|}{2^L}.$
We emphasize that the $\mathcal{B}_{i,\sigma}$ depends on the ID $i$ of $v_i$; it may be the case that the resulting set $\mathcal{B}_{j,\sigma}$ contains a completely different set of neighborhoods for node $v_j$ ($i \ne j$). 
Recall from Lemma~\ref{lem:distinguish} that a node $v_i \in V$ may deduce whether $\sigma = i$ by inspecting its degree and, consequently, the algorithm at $v_i$ may behave differently when $\sigma \ne i$ compared to the case where $v_i$ is either $A$-restricted or $B$-restricted.
That is, for a given partition $\set{A,B}$ of $W$, the message sent by any $A$-restricted node $v_i \in V$ ($i\ne\sigma$) induces a partition $\mathcal{P}_{i,A}$ on $\Phi_A$, whereas the message sent by a $B$-restricted node $v_j$ induces a partition $\mathcal{P}_{j,B}$ on $\Phi_B$. 

The next definition places some additional restrictions on a separated pair with respect to the three partitions $\mathcal{P}_{i,\sigma}$, $\mathcal{P}_{i,A}$, and $\mathcal{P}_{i,B}$, which we will leverage in Section~\ref{sec:simulation}.

\begin{definition}[Indistinguishable Separated Pair] \label{def:indist_sep_pair}
We say that \emph{$S_0$ and $S_1$ form an indistinguishable separated pair for $v_i$}, if the following properties hold:
    \begin{enumerate}
        \item[(i)]$S_0$ and $S_1$ are in the same block of $\mathcal{P}_{i,\sigma}$. \label{item:sigma}
        \item[(ii)] $\phi_A(S_0)$ and $\phi_A(S_1)$ are distinct and in the same block of $\mathcal{P}_{i,A}$, and \label{item:A}
        \item[(iii)] $\phi_B(S_0)$ and $\phi_B(S_1)$ are distinct and in the same block of $\mathcal{P}_{i,B}$. \label{item:B}
        \item[(iv)] $S_0$ and $S_1$ form a separated pair for $v_i$. %
        \label{item:size}
    \end{enumerate}
\end{definition}

\begin{lemma} \label{lem:constant_prob}
Suppose that we obtain the partition $\set{A,B}$  by assigning each vertex in $W$ uniformly at random to either $A$ or $B$.
    Consider any $v_i \in V$.
    Then, with probability at least $\frac{3}{4}$, there exist distinct $S_0,S_1 \in \mathcal{S}$ that form an indistinguishable separated pair for $v_i$. 
\end{lemma}

\begin{proof}
We start by defining additional partitions $\tilde{\mathcal{P}}_{i,A}$ and $\tilde{\mathcal{P}}_{i,B}$ that we obtain from $\mathcal{P}_{i,A}$ and $\mathcal{P}_{i,B}$, respectively, as follows:
For each block $\mathcal{B} \in \mathcal{P}_{i,A}$, we define 
\[
\phi^{-1}_A(\mathcal{B}) = \bigcup_{T \in \mathcal{B}} \{S \in \mathcal{S}  \mid \phi_A(S) = T\}.
\]
Partition $\tilde{\mathcal{P}}_{i,A}$ has exactly the same number of blocks as $\mathcal{P}_{i,A}$, and we obtain $\tilde{\mathcal{B}} \in \tilde{\mathcal{P}}_{i,A}$ from $\mathcal{B} \in \mathcal{P}_{i,A}$ by defining it as the set containing all elements in $\phi^{-1}_A(\mathcal{B})$.
Analogously, we construct $\tilde{\mathcal{P}}_{i,B}$, which will have the same number of blocks as $\mathcal{P}_{i,B}$, and, for each $\mathcal{B}'\in \mathcal{P}_{i,B}$, we get block $\tilde{\mathcal{B}'} \in \tilde{\mathcal{P}}_{i,B}$, which contains all elements in $\phi^{-1}_B(\mathcal{B}') = \bigcup_{T \in \mathcal{B}'} \{S \in \mathcal{S}  \mid \phi_B(S) = T\}$.
Thus, the number of blocks in $\tilde{\mathcal{P}}_{i,A}$, $\tilde{\mathcal{P}}_{i,B}$, and $\mathcal{P}_{i,\sigma}$ is at most $2^L$, and all three of them partition the set $\mathcal{S}$ given by Lemma~\ref{lem:S}.

To show the existence of the sought sets $S_0$ and $S_1$, we need the following combinatorial statement:

\begin{claim}\label{cl:partition}
    Let $D$ be a set of size $s$, and let $\mathcal{P}_1, \ldots, \mathcal{P}_r$ be $r$ partitions of $D$, such that the number of blocks in each of the partitions is at most $n' < s^{{1}/{r+1}}$. 
    Then, for every $i \in [r]$, there exists a block $B_i$ in $\mathcal{P}_i$ that contains at least $s^{1-i/(r+1)}$ elements that are in the same block in each of the partitions $\mathcal{P}_1, \ldots, \mathcal{P}_{i}$.
\end{claim}
\begin{proof}[Proof of Claim~\ref{cl:partition}]
    The case $i=1$ is immediate by the pigeonhole principle. 
    For the inductive step, assume that the claim is true for some $i<r$. That is, there is some block $B_i \in \mathcal{P}_i$ containing at least $s^{1-i/(r+1)}$ elements that are also in the same block in each of the partitions $P_1, \ldots, P_{i}$. 
    Again, by the pigeonhole principle, we conclude that there exists a block $B_{i+1}\in \mathcal{P}_{i+1}$ that has at least $s^{1-(i+1)/(r+1)}$ elements of $B_i$, as required.
\end{proof}

Instantiating Claim~\ref{cl:partition} with $n' = 2^{L} = 2^{o(k)}$ and partitions $\tilde{\mathcal{P}}_{i,A}$, $\tilde{\mathcal{P}}_{i,B}$, and $\mathcal{P}_{i,\sigma}$, it follows that there exists a set $\mathcal{T} \subseteq \mathcal{S}$, such that all sets in $\mathcal{T}$ are in the same block in all three partitions. 
Moreover, \eqref{eq:S} tells us that 
\begin{align*}
|\mathcal{T}|\ge |\mathcal{S}|^{1/4} \ge e^{(k-1)/2}.
\end{align*}
Without loss of generality, assume that $|\mathcal{T}|$ is even, and pair up the sets in $\mathcal{T}$ in some arbitrary way. 
Let $\hat{\mathcal{T}}$ denote this set of pairs, and observe that 
\begin{align}
|\hat{\mathcal{T}}| \ge  e^{(k-1)/2 - \log 2}. \label{eq:that}
\end{align}
Consider any pair $(S_0,S_1) \in \hat{\mathcal{T}}$, and note that we have already shown that Property~(i) of Definition~\ref{def:indist_sep_pair} holds.
To show Properties~(ii) and (iii), suppose that $S_0,S_1 \in \tilde{\mathcal{B}}_A \in \tilde{\mathcal{P}}_{i,A}$.
By construction of $\tilde{\mathcal{P}}_{i,A}$, we know that there exist some $T_0, T_1 \in \mathcal{B}_A$ where $\mathcal{B}_A$ is a block in $\mathcal{P}_{i,A}$, such that $T_0=\phi_A(S_0)$ and $T_1 = \phi_A(S_1)$, and, analogously, it follows that $\phi_B(S_0)$ and $\phi_B(S_1)$ are in a common block of $\mathcal{P}_{i,B}$. 

To complete the proof, we need to bound the probability that $S_0$ and $S_1$ form a separated pair. 
Note that this will also imply $\phi_A(S_0) \ne \phi_A(S_1)$ and  $\phi_B(S_0) \ne \phi_B(S_1)$, by Definition~\ref{def:sep_pair} (see page~\pageref{def:sep_pair}).

Since $|S_0|=|S_1|=2k-1$, it must be true that either $|S_0 \cap A| \ge k$ or $|S_0 \cap B| \ge k$.
Without loss of generality, we can assume the former inequality holds---otherwise, simply exchange $S_0$ and $S_1$. 
Recalling that $S_0,S_1 \in \mathcal{S}$, we know that $|S_0 \cap S_1| \le \epsilon k$, which means that, conditioned on $|S_0 \cap A| \ge k$, there are still at least 
\begin{align}
|S_1 \setminus  S_0| \ge 2k-1-\lceil \epsilon k \rceil
\end{align}
elements of $S_1$, whose random assignment is not fixed by the conditioning.
Let $\ell$ be the largest even integer such that 
$\ell \le 2k-1- \lt\lceil \epsilon k  \rt\rceil,$ and note that
\begin{align}
k(2- \epsilon)-2  \le \ell \le k(2-\epsilon)-1.  \label{eq:ell}
\end{align}
Fix some order of the elements in $S_1  \setminus S_0$, and define $Z_i$ ($i \in [\ell]$) to be the indicator random variable that is $1$ if and only if the $i$-th element is in $B$.
Let $Z = \sum_{i=1}^{\ell}Z_i$.
Our goal is to show that $|(S_1 \setminus  S_0) \cap B| \ge Z \ge k$.
In the worst case, all $\lt\lceil \epsilon k \rt\rceil$ elements in $S_1 \cap S_0$ lie in $A$, and hence $\EE\lt[ Z \ \md|\ |S_0 \cap A| \ge k \rt] < k$.
Therefore, we need to make use of the following anti-concentration result to show that there is a sufficiently large probability for $Z\ge k$: 
\begin{lemma}[Proposition~7.3.2 in \cite{matouvsek2001probabilistic}] \label{lem:helper}
    Let $\ell$ be even, $X_1, \ldots, X_{\ell}$ be independent random variables, where each takes the values of $0$ and $1$ with probability $\frac{1}{2}.$ 
    Let $X = X_1 + X_2 + \ldots X_{\ell}$. 
    For any integer $t \in [\ell/8]$, it holds that
    $\Pr \left[X \geq \frac{\ell}{2}+ t \right] \geq \frac{1}{15} e^{-16t^2/\ell}.$
\end{lemma}
To apply Lemma~\ref{lem:helper}, define $t$ to be an integer such that 
\begin{align}
\frac{ \epsilon k +1}{2} \le t \le \frac{ \epsilon k + 4}{2}, \label{eq:t}
\end{align}
and note that $t < \ell/8$ holds, as long as $\epsilon \leq \frac{2}{5}- \frac{18}{5k}$ since $k = \omega(1)$.
By Lemma~\ref{lem:helper}, we have 

\begin{align}
  \Pr\Big[ Z \ge k \ \Big|\ |S_0 \cap A| \ge k\Big]
  &\ge
  \Pr\Big[ Z \ge \frac{\ell}{2} + t \ \Big|\ |S_0 \cap A| \ge k\Big]  \notag\\ 
  \ann{by Lem.~\ref{lem:helper}, and $t<\ell/8$}
  \notag 
  &\ge
  \frac{1}{15}\exp\lt( -
  \frac{\ell}{4}\rt)
  \notag\\ 
  \ann{by \eqref{eq:ell}} 
  \notag
  &\ge
  \frac{1}{15}\exp\lt( -
  \frac{k(2-\epsilon)}{4}\rt)
  \notag
\end{align} 

Recalling \eqref{eq:that}, it follows that the probability that none of the pairs in $\hat{\mathcal{T}}$ is a separated pair is at most  
\begin{align}
    \lt(1- \frac{e^{-\frac{k(2-\epsilon)}{4}
    }}{15}\rt) ^{e^{\frac{k-1}{2}-\log 2}}
  \notag
  \le
    \exp \lt (-
    \frac{e^{-\frac{k(2-\epsilon)}{4}+\frac{k-1}{2}-\log 2}}{15}
    \rt )
  \notag 
  \le
  \exp \lt ( - \frac{e^{\frac{\epsilon k}{4}-2}}
  {15} \rt )
\end{align}
which is at most $\frac{1}{4}$ for sufficiently large $n$. 

It follows that one of the pairs in $\hat{\mathcal{T}}$, all of which satisfy Properties~(i)-(iii), is a separated pair with probability at least $\frac{3}{4}$, which shows Property~(iv) and completes the proof of Lemma~\ref{lem:constant_prob}.
\end{proof}

Recall that $|V|=\Theta\lt( n \rt)$.
For a randomly chosen partition $\set{A,B}$ of $W$, Lemma~\ref{lem:constant_prob} implies that a constant fraction of the nodes $v_i \in V$ has indistinguishable separated pairs in expectation.
Therefore, there must be some concrete choice $\set{A^*,B^*}$ for $\set{A,B}$ that ensures these properties.
We are now ready to summarize the main result of this section while recalling the assumed upper bound on the message size: 

\begin{corollary} \label{cor:exists}
For a given $k$-edge connectivity algorithm, let $V_{\text{good}} \subseteq V$ be the set of nodes for which there exist indistinguishable separated pairs (see Def.~\ref{def:indist_sep_pair} on page~\pageref{def:indist_sep_pair}). 
If each node sends at most $o(k)$ bits, then there exists a partition $\set{A^*,B^*}$ of $W$ such that $|V_{\text{good}}| \ge \frac{3}{4}|V|$.
\end{corollary}

%% file: overlap.tex
\section{The $\overlap$ Problem} \label{sec:overlap}

We now introduce a communication complexity problem in a 3-party model with simultaneous messages, whose hardness turns out to be crucial for obtaining a lower bound for  deciding $k$-edge connectivity in the sketching model. 

In the $\overlap_{m,s}$ problem, we are given two positive integers $m$ and $s$, where $s \le \lceil m /2 \rceil$, and there are three players Alice, Bob, and Charlie. %
We focus on one-way communication protocols, where Alice and Bob each send a single message to Charlie, who outputs the result.
Alice's input is a ternary vector of length $m$, denoted by $X \in \set{0,1,\perp}^m$, and, similarly, Bob gets a vector $Y \in \set{0,1,\perp}^m$. 
We use $\supp(Z)$ to denote the support of a vector $Z$, i.e., the set of indices in $[m]$, for which $Z_i \ne \perp$.
Charlie does not know anything about $X$ and $Y$, except $\supp(X)$ and $\supp(Y)$.

\begin{definition}[Valid Instance for $\overlap_{m,s}$] \label{def:valid}
The tuple $(X,Y)$ is a \emph{valid instance for $\overlap_{m,s}$}, if $|X| = |Y| = s$ and the following properties hold:
\begin{enumerate} 
\item[(P1)] There is exactly one index $\sigma \in [m]$ such that $X_\sigma \oplus Y_\sigma = 1$, where $\oplus$ denotes the XOR operator.
\item[(P2)] For every $i \in [m] \setminus  \set{\sigma}$, it holds that $X_i = \perp$ or $Y_i = \perp$.
\end{enumerate}
\end{definition}
In conjunction, Properties~(P1) and (P2) ensure that $|\supp(X) \cap \supp(Y)| =1$, and thus Charlie can deduce the value of $\sigma$ directly from his input.
To solve the problem, Alice and Bob each send a single message to Charlie, who must output ``yes'' if $X_\sigma = 0$ and $Y_\sigma = 1$, and ``no'' otherwise.
See \onlyLong{Figure~\ref{fig:overlap}}\onlyShort{Figure~2 in the attached full paper} for an example.
\onlyLong{
\begin{figure}[t]
  \centering
\includegraphics[scale=1.0]{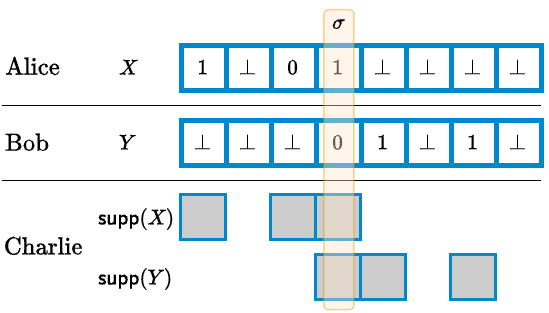}
\caption{\small
An instance of $\overlap_{8,3}$. Since $(X_\sigma,Y_\sigma)=(1,0)$, Charlie needs to answer ``no''.}
\label{fig:overlap}
\end{figure}
}

\subsection{Relationship to Other Problems in Communication Complexity.} \label{sec:overlap_rel}
In \cite{ashvinkumar2023evaluating}, the authors define the Leave-One Index problem (Problem 3 in \cite{ashvinkumar2023evaluating}), where there are three players Alice, Bob, and Charlie. 
Alice and Bob hold the same string $x \in \set{0,1}^N$ and Charlie knows an index $i^* \in [N]$. 
In addition, Alice has a set $S \subseteq [N]$ and Bob has a set $T \subseteq [N]$ with the property that $S \cup T = [N] \setminus  \set{i^*}$, and the goal is that the players output $x_{i^*}$, whereby the order of communication is Alice, Bob, and then Charlie. 
We point out that Leave-One Index lacks an important requirement of the $\overlap$ problem, namely that Charlie also knows the exact support sets (not just their size) of the sets $S$ and $T$ given to Alice and Bob, which is crucial for the correctness of our simulation. Thus, it is unclear whether there exists a reduction between these problems.  

The uniqueness of the shared index $\sigma$ may give the impression, at first, that $\overlap$ is related to a (promise) variant of computing the set intersection. 
However, we doubt the existence of a straightforward reduction, as, for the $\overlap$ problem, the input vectors of Alice and Bob are ternary, and Charlie knows the (single) intersecting index as well as the support of the two input sets.

\subsection{Hardness of the $\overlap$ Problem} \label{sec:overlap_lb}

We now show that either Alice or Bob needs to send a message of size $\Omega\lt( m \rt)$ bits to Charlie in the worst case.

\begin{reptheorem}{thm:overlap}
\thmOverlap
\end{reptheorem}
\begin{proof}
Consider any protocol for $\overlap_{m,s}$ and let $C$ be the worst case length of any message sent.
Assume towards a contradiction that 
\begin{align}
C \le \lt\lfloor \frac{m}{6} \rt\rfloor -1.  \label{eq:C}
\end{align}

Suppose that we fix Alice's support to be a set $I$ of $C+f$ indices from $[m]$, where $f$ is an integer parameter described below. 
We can partitions her $2^{C+f}$ possible input vectors into $2^{C}$ blocks such that the $i$-th block contains all inputs on which Alice sends the number $i$ to Charlie. 
Clearly, the largest block $\mathcal{B}$ in this partition contains at least $2^{f}$ inputs. 
Observe that there must exist a set $F_\text{Alice}(I) \subseteq I$ of $f$ indices in $I$ such that, for each $i \in F_\text{Alice}(I)$, there are inputs $X,\hat{X} \in \mathcal{B}$ such that $X_i \ne \hat{X}_i$, i.e., Alice sends the same message for two inputs in which the bit at index $i$ is $0$ and $1$, respectively.
We say that index $i$ is \emph{flipped for $I$}, and we refer to $F_\text{Alice}(I)$ as \emph{Alice's flipped indices for support $I$}.
Similarly, we can identify a set of flipped indices $F_{\text{Bob}}(I)$ for Bob by partitioning all possible input vectors with support $I$ into blocks. 
Note that $F_{\text{Alice}}(I)$ and $F_{\text{Bob}}(I)$ need not be the same since Alice and Bob can execute different algorithms. 

A crucial observation is that the algorithm fails, if there exist support sets $I_1$ and $I_2$ such that 
\begin{enumerate} 
\item[(a)] $|I_1 \cap I_2| = 1$, and 
\item[(b)] $\lt( I_1 \cap I_2 \rt) \subseteq \lt( F_{\text{Alice}}(I_1) \cap F_{\text{Bob}}(I_2) \rt)$.
\end{enumerate}
These conditions imply that the important index $\sigma \in I_1 \cap I_2$ is flipped for Alice's as well as Bob's support. 
In more detail, this means that there are two inputs $X$ and $\hat{X}$ for Alice that are in the same block $\mathcal{B}_{\text{Alice}}$, where $X_\sigma\ne \hat{X}_\sigma$ and Alice sends the same message $\pi(\mathcal{B}_{\text{Alice}})$, and there are two inputs $Y$ and $\hat{Y}$ in the same block $\mathcal{B}_{\text{Bob}}$, where $Y_\sigma\ne\hat{Y}_\sigma$, for which Bob  sends $\pi(\mathcal{B}_{\text{Bob}})$, thus causing Charlie to fail in one of the (two) valid input combinations.

We show that if an algorithm manages to avoid (a) and (b) for all possible pairs of support sets, then it must send at least $C = \Omega\lt( m \rt)$ bits.
Let $\mathcal{I} = \binom{[m]}{s}$ be the set family of all possible $s$-element support sets for the inputs of Alice and Bob.
For each $i \in [m]$, define 
\[
\mathcal{R}_i^{\text{Alice}} = \Set{ I \in \mathcal{I}  \mid  i \in F_\text{Alice}(I) },
\] 
which means that $\mathcal{R}_i^{\text{Alice}}$ contains every support set 
$I$ where $i$ is a flipped index for $I$ of Alice, and we define  $\mathcal{R}_i^{\text{Bob}}$ analogously. 

For the rest of the proof, if a statement holds for both 
$\mathcal{R}_i^{\text{Alice}}$ and $\mathcal{R}_i^{\text{Bob}}$, we simply omit the superscript and write $\mathcal{R}_i$ instead.  
Recall that each support $I \in \mathcal{I}$ contains at least $f$ flipped indices, i.e., $|F_\text{Alice}(I)| \ge f$ and $|F_\text{Bob}(I)| \ge f$. 
This tells us that $I$ must be a member of at least $f$ of the $\mathcal{R}_i$ set families, and thus
\begin{align}\label{eq:sumR}
    \sum_{i=1}^{m} |\mathcal{R}_i^{\text{Alice}}| + |\mathcal{R}_i^{\text{Bob}}| \ge 2f\binom{m}{r}.
\end{align}

Consider $I_1 \in \mathcal{R}_i^{\text{Alice}}$ and $I_2 \in \mathcal{R}_i^{\text{Bob}}$.
The only way to avoid (a) and (b) is if $I_1$ and $I_2$ are not valid choices for the support of Alice's input $X$ and Bob's input $Y$, respectively (see Def.~\ref{def:valid}), which is the case if $|I_1 \cap I_2| \geq 2$; (recall that $I_1 \cap I_2 =\emptyset$ is impossible since $i \in I_1\cap I_2$).
For any $\mathcal{R}_i$, we can obtain a new set family $\hat{\mathcal{R}}_i$ by removing the index $i$ from each set in $\mathcal{R}_i$, i.e.,
\begin{align}
\hat{\mathcal{R}}_i = \Set{ I' \in \binom{[m]}{s-1}  \mid  \exists I \in \mathcal{R}_i\colon I = I' \cup \set{i}}.
\notag
\end{align}
The correctness of the assumed protocol implies the following:
\begin{fact}
Set families $\hat{\mathcal{R}}_i^{\text{Alice}}$ and $\hat{\mathcal{R}}_i^{\text{Bob}}$ form a cross-intersecting family, which means that
\begin{align}
\forall I_1' \in \hat{\mathcal{R}}_i^{\text{Alice}}\ \forall I_2' \in \hat{\mathcal{R}}_i^{\text{Bob}}\colon |I_1' \cap I_2'| \geq 1. 
\end{align}
\end{fact}
Next, we employ an upper bound on the sum of the cardinalities of any two cross-intersecting set families: 
\begin{lemma}[\cite{hilton1967some}] \label{lem:cross}
Let $\mathcal{A}$ and $\mathcal{B}$ be such that each is a family of $\ell$-element sets on the same underlying set of size $p$ and assume that $\mathcal{A}$ and $\mathcal{B}$ are cross-intersecting. If $p \geq 2\ell$, then  $|\mathcal{A}|+|\mathcal{B}| \leq 1+ \binom{p}{\ell} - \binom{p-\ell}{\ell}.$
\end{lemma}
Choosing $f = \lt\lfloor \frac{m}{3} \rt\rfloor +1$ and recalling \eqref{eq:C}, implies that 
\begin{align}
s = C + f  \le \frac{m}{2}, \label{eq:scf}
\end{align}
and therefore $2(s-1) \leq m-1$.
Thus, we can apply Lemma~\ref{lem:cross} with parameters $p=m-1$ and $\ell = s-1$ to obtain
\begin{align}\label{eq:sumR'}
    |\hat{\mathcal{R}}_i^{\text{Alice}}| +  |\hat{\mathcal{R}}_i^{\text{Bob}}| \leq 1+ \binom{m-1}{s-1}-\binom{m-s}{s-1}.
\end{align}
Since $|\hat{\mathcal{R}}_i| = |\mathcal{R}_i|$, it follows from \eqref{eq:sumR} and \eqref{eq:sumR'} that 
\begin{align} \label{eq:combine}
    2f\binom{m}{s} \leq 1+ \binom{m-1}{s-1}-\binom{m-s}{s-1}.
\end{align}
However, it also holds that
\begin{align}
    2f\binom{m}{s} 
    & = 2\lt( \lt\lfloor \frac{m}{3}  \rt\rfloor +1 \rt) \frac{m}{s} \binom{m-1}{s-1}\notag \\
    \ann{by \eqref{eq:scf}}
    & \ge 4\lt( \lt( \frac{m}{3} - 1 \rt) + 1 \rt)  \binom{m-1}{s-1}\notag \\
    & = \frac{4}{3} m \binom{m-1}{s-1} \notag \\
    & > m \left(1+ \binom{m-1}{s-1}\right), \notag 
\end{align}
which provides a contradiction to \eqref{eq:combine}, thus completing the proof of Theorem~\ref{thm:overlap}.
\end{proof}

%% file: overlap_algo.tex
\subsection{A Deterministic Algorithm for $\overlap$} \label{sec:algo}

The argument developed for our lower bound proof in Section~\ref{sec:overlap_lb} inspires the design of a simple deterministic algorithm that allows Alice and Bob to save a single bit by sending messages of length $s-1$ under certain conditions.
As this is not needed for proving our main result, we relegate the details to \onlyLong{Appendix~\ref{app:algo}.}\onlyShort{the appendix of the attached full paper.}

\newcommand{\thmAlgo}{
If $\lceil \frac{m}{3} \rceil < s$, then there exists a deterministic one-way protocol for $\overlap_{m,s}$ in the simultaneous $3$-party model, such that Alice and Bob send at most $s-1$ bits to Charlie.
}
\begin{theorem} \label{thm:algo}
\thmAlgo
\end{theorem}

%% file: simulation.tex
\section{Simulation} \label{sec:simulation}

In this section, we show how Alice, Bob, and Charlie can jointly simulate a given $k$-edge connectivity algorithm $\mathcal{A}^{\text{$k$-conn}}$ in the 3-party model to solve the $\overlap_{m,s}$ problem, for some given integers $m$ and $s$, assuming that the sketches are sufficiently small for Corollary~\ref{cor:exists} to hold. %

We now describe the details of the simulation and how the players can create  the lower bound graph (see Sec.~\ref{sec:lb_graph}), given a valid instance of $\overlap_{m,s}$.
First, each player locally computes 
\begin{align}
n = 2\lt\lceil {4m}/{3} \rt\rceil, \label{eq:n} 
\end{align}
which will correspond to the number of nodes in the graph $G \in \mathcal{G}_{k,n}$ that the players construct. 
According to Corollary~\ref{cor:exists}, there is a partitioning of $W$ into $\set{A^*,B^*}$, as well as a set $V_{\text{good}}$ of size at least $m$ vertices, with the property that there exist indistinguishable separated pairs (see Def.~\ref{def:indist_sep_pair}), for each of the nodes in $V_{\text{good}}$.
We use $V_{\text{good}}^{\leq m}$ to denote the first $m$ nodes in $V_{\text{good}}$, ordered by IDs.
Note that the players know, in advance, the IDs in the sets $V$, $V_{\text{good}}$, $A^*$, $B^*$, as these are fixed for all possible lower bound graphs in $\mathcal{G}_{k,n}$.
Furthermore, they can also pre-compute the indistinguishable separated pair $(S_0^i,S_1^i)$ , for every $v_i \in V_{\text{good}}^{\leq m}$, as these are fully determined by fixing $k$, $n$, and the algorithm at hand.

\onlyLong{
\begin{figure}[t]
\centering
\begin{subfigure}[t]{0.3\textwidth}
\includegraphics[scale=0.9]{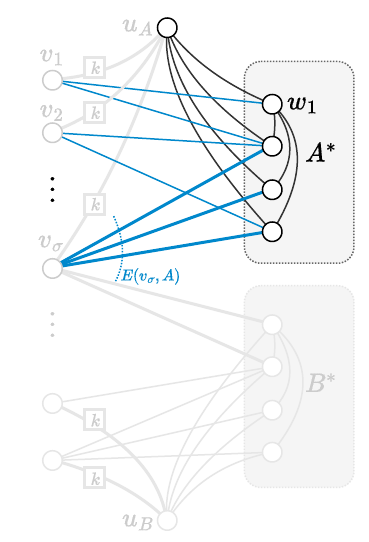}
\caption{Alice simulates all nodes in $A^*$.}
\label{fig:alice}
\end{subfigure}
\hspace{2mm}
\begin{subfigure}[t]{0.3\textwidth}
\includegraphics[scale=0.9]{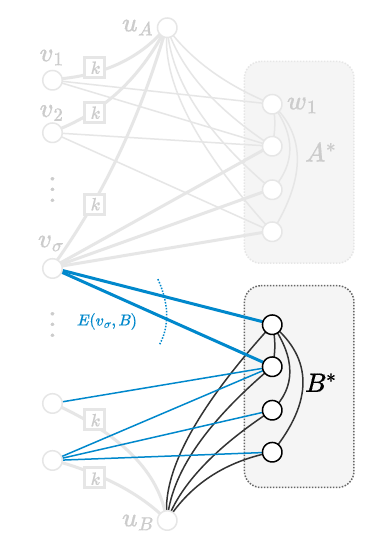}
\caption{Bob simulates all nodes in $B^*$.}
\label{fig:bob}
\end{subfigure}
\hspace{2mm}
\begin{subfigure}[t]{0.3\textwidth}
\includegraphics[scale=0.9]{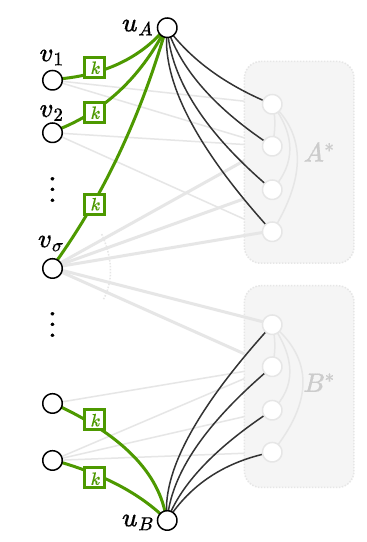}
\caption{Charlie simulates $u_A$ and $u_B$, as well as all nodes in $V$, even though he only knows a subset of their incident edges.}
\label{fig:charlie}
\end{subfigure}
  \centering
\caption{
Overview of the Simulation. Each player simulates a certain subset of the nodes and their incident edges. Note that the edges of $v_\sigma$ are ``split'' between Alice and Bob.
}
\label{fig:sim}
\end{figure}
}
\subsubsection*{Alice's and Bob's Simulation:}
Alice simulates all nodes in $A^*$, whereas Bob is responsible for the nodes in $B^*$. 
See \onlyLong{Figures~\ref{fig:alice} and \ref{fig:bob}}\onlyShort{Figures~3a and 3b in the attached full paper} for an example. 
First, Alice creates a clique on $A^*$ and connects $u_A$ to each node in this clique.
Then, for each $i \in \supp(X)$ where $v_i \in V_{\text{good}}^{\leq m}$, Alice inspects the input $X_i$ and adds edges incident to the nodes in $A^*$ as follows: 
\begin{itemize} 
\item If $X_i = 0$, then she includes the edges between $v_i$ and the nodes in $S_1^i \cap A^*$. 
\item On the other hand, if $X_i = 1$, then Alice adds all edges connecting $v_i$ to $S_0^i \cap A^*$.
\end{itemize}

Bob adds a clique on $B^*$ and, analogously to Alice, includes an edge between $u_B$ and each node in $B^*$. Next, we describe how Bob creates the edges between $B^*$ and node $v_j \in V_{\text{good}}^{\leq m}$, for each $j \in \supp(Y)$:
\begin{itemize} 
\item If $Y_j = 0$, then he adds edges between $v_j$ and $S_0^j \cap B^*$. 
\item If $Y_j = 1$, he connects $v_j$ to each node in $S_1^j \cap B^*$.
\end{itemize} 
Finally, Alice and Bob locally execute algorithm $\mathcal{A}^{\text{$k$-conn}}$ to compute the sketches produced by their simulated nodes in $A^*$ and $B^*$, respectively, which they send to Charlie.

\subsubsection*{Charlie's Simulation:}
Charlie is responsible for computing the sketches for $u_A$, $u_B$, and all nodes in $V$, and he will also simulate the referee. 
To this end, he creates edges between $u_A$ and each node in $A^*$---recall that these edges were also created by Alice for simulating the nodes in $A^*$---and he also connects $u_A$ to every $v_i \in V \setminus \supp(Y)$ via $k$ parallel edges.
Note that $V \setminus  \supp(Y)$ contains every $i$ such that $i \in \supp(X)$ or $i \notin \supp(X) \cup \supp(Y)$, whereby the latter case is equivalent to $X_i=Y_i=\perp$. 
Charlie also includes $k$ parallel edges connecting $u_B$ to every $v_j$ where $j \in \supp(Y) \setminus \set{\sigma}$. 
Similarly, he adds edges between $u_B$ and each node in $B^*$. \onlyLong{See Figure~\ref{fig:charlie}.}\onlyShort{See Figure~3c in the attached full paper.}

Next, we describe how Charlie simulates the nodes in $V$:
\begin{itemize} 

\item 
First, consider the crucial node $v_i \in V_{\text{good}}^{\leq m}$, where $\sigma = i$:
Since Charlie does not know whether Alice and Bob used $S_0^i$ or $S_1^i$ for creating the edges between their nodes in $A^* \cup B^*$ and $v_\sigma$, he cannot hope to faithfully recreate the entire neighborhood of $v_i$. 
However, what he can do instead is to directly compute the sketch that algorithm $\mathcal{A}^{\text{$k$-conn}}$ produces for node $v_i$, in the case that $\sigma=i$.
To see why this is true, recall that $S_0^i$ or $S_1^i$ form an indistinguishable separated pair, and Charlie knows the block $\mathcal{B}_i$ in partition $\mathcal{P}_{i,\sigma}$ that contains $S_0^\sigma$ as well as $S_1^\sigma$.
This ensures that $v_i$ sends the same sketch $\pi(\mathcal{B}_i)$ given either neighborhood. 
\item 
For any node $v_j \in V_{\text{good}}^{\leq m}$ ($j \ne \sigma$),
Charlie proceeds analogously as in the case $j = \sigma$, with the only difference being that he no longer uses $\mathcal{P}_{j,\sigma}$ for identifying the common block that contains $S_0^j$ and $S_1^j$ (and deriving $v_j$'s sketch). 
Instead, if $v_j$ is $A$-restricted (i.e., $j \in \supp(X) \setminus  \supp(Y)$), then he uses $\tilde{\mathcal{P}}_{j,A^*}$ to find the common block that contains $S_0^j$ and $S_1^j$ and uses the associated block in $\mathcal{P}_{j,A^*}$ (see Page~\pageref{lem:constant_prob}) to derive $v_j$'s sketch. 
Otherwise, $v_j$ must be $B$-restricted ($j \in \supp(Y) \setminus  \supp(X)$), prompting him to find the common block in $\tilde{\mathcal{P}}_{j,B^*}$, and compute $v_j$'s sketch accordingly. 

\item 
Charlie can directly simulate the remaining nodes in $V$, i.e., every 
\[
v_r \in (V \setminus V_{\text{good}}^{\leq m}) \cup \Set{ v_i \in V_{\text{good}}^{\leq m} \ \md|\ i \notin \lt( \supp(X) \cup \supp(Y) \rt)},
\]
since the only edges added for $v_r$ are $k$ parallel edges to $u_A$.
\end{itemize}
Charlie simulates that the referee receives the sketches sent by the nodes in $V$, $u_A$, and $u_B$, as well as the messages that it received from Alice and Bob for the nodes in $A^*$ and $B^*$, which is sufficient for invoking $\mathcal{A}^{\text{$k$-conn}}$ and obtaining the decision by the referee. 
Charlie answers ``yes'' if and only if the decision was that the graph was $k$-edge connected.  

\subsection{Correctness and Complexity Bounds of the Simulation} \label{sec:sim_analysis}

As made evident by the above description, our simulation constructs only certain graphs in $\mathcal{G}_{k,n}$. 
In particular, these are graphs where all edges in the cut $E(V,W)$ are determined by separating pairs.  
We start by formalizing this restriction:

\begin{definition} \label{def:compat}
Consider any graph $G \in \mathcal{G}_{k,n}$ and recall the existence of set $V_{\text{good}}$, the partition $(A^*, B^*)$, and the separating pairs $(S_0^i, S_1^i)$ for each $v_i \in V^{\leq m}_{\text{good}}$, guaranteed by applying Corollary~\ref{cor:exists} to algorithm $\mathcal{A}^{\text{$k$-conn}}$.
We say that $G$ is \emph{compatible} with a valid instance $(X,Y)$ of $\overlap_{m,s}$ if $|V_{\text{good}}|\ge m$ and, for every $i \in [m]$: 
\begin{enumerate} 
\item If $X_i=0$, then $E(v_i,A^*)=A^* \cap S_1^i$. 
\item If $X_i=1$, then $E(v_i,A^*)=A^* \cap S_0^i$. 
\item If $Y_i=0$, then $E(v_i,B^*)=B^* \cap S_0^i$. 
\item If $Y_i=1$, then $E(v_i,B^*)=B^* \cap S_1^i$. 
\item If $i \in \supp(X)$ or $i \notin (\supp(X) \cup \supp(Y))$, then $v_i$ has $k$ edges to $u_A$. 
\item If $i \in \supp(Y) \setminus \supp(X)$, then $v_i$ has $k$ edges to $u_B$. 
\end{enumerate}
\end{definition}

\begin{lemma} \label{lem:compat}
Let $G$ be a graph that is compatible with an instance $I = (X, Y)$ of $\overlap$.
It holds that $G$ is $k$-edge connected if and only if the solution to $I$ is ``yes'', i.e., $X_\sigma=0$ and $Y_\sigma=1$.
\end{lemma}
\onlyLong{
\begin{proof}
Since $G$ is a compatible graph, we know that $G \in \mathcal{G}_{k,n}$.
Thus, it follows from Lemma~\ref{lem:lb_graph} that $G$ is $k$-edge connected if and only condition (C1) holds, i.e., $|E(v_\sigma,A^*)| \le k-1$ and $|E(v_\sigma,B^*)| \ge k$.
Since $(S_0^\sigma,S_1^\sigma)$ form a separated pair, we know from Definition~\ref{def:sep_pair} that this is satisfied for $S_1^\sigma$.
According to Definition~\ref{def:compat}, we use $S^{\sigma}_1$ to determine the neighbors of $v_\sigma$ in $A^*\cup B^*$ if and only if $X_\sigma=0$ and $Y_\sigma=1$, which proves the correspondence between $k$-edge connectivity and the solution to $\overlap$.
\end{proof}
}

Equipped with Definition~\ref{def:compat}, we \onlyShort{show in the attached full paper}\onlyLong{are ready to show} that the sketches produced by the simulation indeed correspond to an actual execution of $\mathcal{A}^{\text{$k$-conn}}$ on a suitable graph in $\mathcal{G}_{k,n}$.

\begin{lemma} \label{lem:same_messages}
Suppose that  algorithm $\mathcal{A}^{\text{$k$-conn}}$ has a maximum sketch length of $L=o(k)$ bits.
Then, Charlie provides as input to the referee in the simulation the same set of sketches that the referee receives when executing algorithm $\mathcal{A}^{\text{$k$-conn}}$ on a graph compatible with the instance $(X,Y)$. 
\end{lemma}
\onlyLong{
\begin{proof}
Since the players construct a graph $G \in \mathcal{G}_{k,n}$, where $n = 2\lt\lceil \frac{4m}{3} \rt\rceil$, it follows from \eqref{eq:sizeV} (on page~\pageref{eq:sizeV}) that $|V| \ge \lt\lceil \frac{4m}{3} \rt\rceil$.
According to Corollary~\ref{cor:exists}, this implies that the set $V_{\text{good}}$ must have a size of $c\cdot m$ for some constant $c\ge 1$, which ensures that the set $V_{\text{good}}^{\leq m} \subseteq V_{\text{good}}$ of size $m$ exists. 
Recall that all players can compute $V_{\text{good}}$, since this is fully determined by the algorithm at hand (i.e., $\mathcal{A}^{\text{$k$-conn}}$) and the parameters $n$, $k$, and $L$. 
We now argue that the simulation computes the correct sketch for every type of node in graph $G$:

\begin{itemize} 
\item Nodes $u_A$ and $u_B$: These nodes are only simulated by Charlie who knows $\supp(X)$, $\supp(Y)$, and thus also $\sigma$.
By the description of the simulation, Charlie creates their incident edges in a way that satisfies Points~5 and 6 of Definition~\ref{def:compat}.

\item Node $w_j \in A^* \cup B^*$: Without loss of generality, assume that $w_j \in A^*$; the case $w_j \in B^*$ is symmetric. In the execution on the compatible graph $G$, the edges in the cut $E(V,w_j)$ are fully determined by the separating pairs $(S_0^1,S_1^1),\ldots,(S_0^m,S_1^m)$, which corresponds to how Alice creates the edges incident to $w_j$ in the simulation. Moreover, Alice will also add an edge $\set{u_A,w_j}$, which ensures the same sketch for $w_j$ as in the actual execution.

\item Node $v_i \in V_{\text{good}}^{\leq m}$: According to the simulation, Charlie will select one of the blocks $\mathcal{B}_\sigma,\tilde{\mathcal{B}}_A, \tilde{\mathcal{B}}_B$, where $\mathcal{B}_\sigma \in P_{i,\sigma}$, $\tilde{\mathcal{B}}_A \in \tilde{P}_{i, A^*}$ and $\tilde{\mathcal{B}}_B \in \tilde{P}_{i,B^*}$, with the property that each of the blocks contains $(S_0^i, S_1^i)$. Recall that these blocks exist due to Corollary \ref{cor:exists}. 
Let $\mathcal{B}_A \in P_{i, A^*}$ be the associated block  of $\tilde{\mathcal{B}}_A$, and $\mathcal{B}_B \in P_{i, B^*}$ be the associated block  of $\tilde{\mathcal{B}}_B$, see Page~\pageref{lem:constant_prob}.
If $i = \sigma$, then he computes $\pi( \mathcal{B}_\sigma)$ where $\pi( \mathcal{B}_\sigma)$ corresponds to the message sent by $v_i$ for the inputs in $\mathcal{B}_\sigma$; otherwise, if $i \in \supp(X) \setminus \supp(Y)$, then he computes $\pi(\mathcal{B}_A)$. Finally, in the case that $i \in \supp(Y) \setminus \supp(X)$, he computes $\pi(\mathcal{B}_B)$ instead. A crucial observation is that, even though Charlie does not know whether Alice used $S_0^i \cap A^*$ or $S_1^i \cap A^*$ to create the edges in $E(v_i,A^*)$, this does not prevent him from computing the correct sketch for $v_i$, since $S_0^i$ and $S_1^i$ are guaranteed to be in the same block in each of the partitions; a similar argument applies to the edges in $E(v_i,B^*)$. 
\item Node $v_i \in V \setminus  V_{\text{good}}^{\leq m}$: 
These nodes do not have any edges to $A^* \cup B^*$.
Instead, Charlie only creates $k$ parallel edges to $u_A$, which matches the neighborhood of $v_i$ in the compatible graph $G$. \qedhere
\end{itemize}
\end{proof}
}
\begin{lemma} \label{lem:simulation}
Suppose there exists a deterministic $k$-edge connectivity algorithm, with the property that every node sends a sketch of length at most $L=o\lt( k  \rt)$ bits, when executing on $\mathcal{G}_{k,n}$.
Then, $\overlap_{m,s}$ has a deterministic communication complexity of $o\lt( k\cdot \sqrt{m} \rt)$ bits. 
\end{lemma} 
\onlyLong{
\begin{proof}
By Lemma~\ref{lem:same_messages}, Charlie invokes the referee with the same set of sketches as in the execution of algorithm $\mathcal{A}^{\text{$k$-conn}}$ on the compatible graph $G \in \mathcal{G}_{k,n}$.
Thus, correctness follows directly from Lemma~\ref{lem:compat}.

Next, we show the claimed bound on the communication complexity of $\overlap_{m,s}$.
By \eqref{eq:n} (see page~\pageref{eq:n}), we know that the size $n$ of $G$ is chosen such that $n=\Theta\lt( m \rt)$.
Since every node in $A \cup B$ sends a sketch of at most $L$ bits in the simulation, and these are all the nodes simulated by Alice and Bob, 
sending these concatenated sketches to Charlie requires at most
\begin{align}
O\lt( |A \cup B| \cdot L \rt) = O\lt( |W| \cdot L \rt) = O\lt( L \cdot \sqrt{n} \rt) 
=o\lt( k \cdot \sqrt{m} \rt)
\end{align}
bits, where we have used \eqref{eq:sizeW} in the second-last equality.
\end{proof}
}

\section{Combining the Pieces} \label{sec:combining}
We now use the results from the previous sections to prove Theorem~\ref{thm:main}.

\begin{reptheorem}{thm:main}
\thmMain
\end{reptheorem}
Let $L$ be the maximum length of a sketch. 
Suppose that every node in $V$ sends at most $o\lt( k \rt)$ bits, 
and consider an instance of $\overlap_{m,s}$.
Lemma~\ref{lem:simulation} tells us that we can solve an instance of $\overlap_{m,s}$ if Alice and Bob send a message of $o\lt( k \cdot \sqrt{m} \rt)$ bits.
Since $k \le O\lt( \sqrt{n} \rt) = O\lt( \sqrt{m} \rt)$,  the messages sent is thus $o(m)$ bits. 
We obtain a contradiction to 
Theorem~\ref{thm:overlap} which states that any deterministic algorithm for $\overlap_{m,s}$ must send $\Omega\lt( m \rt)$ bits in the worst case.

%% file: conclusion.tex
\section{Conclusion and Open Problems} \label{sec:conclusion}

A problem left open by our work is the complexity of \emph{randomized} sketching algorithms for deciding $k$-edge connectivity. 
As outlined in Section~\ref{sec:intro}, the existing lower bound approach for graph connectivity does not appear to be amenable to a straightforward generalization.
 This suggests that either new technical ideas are needed for proving such a bound, or perhaps there is indeed a more efficient randomized algorithm:

\begin{problem}
What is the sketch length of deciding $k$-edge connectivity, if we allow a small probability of error? 
\end{problem}

Our result implies that any deterministic sketching algorithm must send sketches of near-linear (in $k$) bits in the worst case. 
While the AGM sketches~\cite{AGM-soda12} yield an upper bound of $O(k\log^{3} n)$ bits, currently there is no non-trivial deterministic upper bound:  

\begin{problem}
Is there a deterministic $k$-edge connectivity algorithm that matches the best known randomized upper bound of $O\lt( k\log^3n \rt)$ bits on the sketch length?
\end{problem}

Finally, we point out that our techniques do not shed light on the most fundamental open problem in this setting, namely (1-edge) connectivity:

\begin{problem}
Can we solve connectivity deterministically in the distributed sketching model if messages are of length $o(n)$?
\end{problem}

%% file: app_lemS.tex
\section{Proof of Lemma~\ref{lem:S}} \label{app:lem:S}
\begin{replemma}{lem:S} 
\lemS
\end{replemma}
\begin{proof}
The proof is similar to Lemma~6 in \cite{kapralov2017optimal}.
Fix $N=\lt\lfloor \sqrt{\lt( \frac{\epsilon|W|}{2e d} \rt)^{\epsilon d/2}-1} \rt\rfloor$. We choose $S_1,\dots,S_{N}$ independently and uniformly at random from ${W \choose d}$.
We first bound the expected size of the overlap of two given subsets. 
Consider two distinct indices $i,j \le N$, and define indicator random variables $X_1,\dots,X_{d}$ such that $X_k=1$ if and only if the $k$-th  element of $S_j$ intersects with $S_i$, which happens with probability $\frac{d}{|W|}$, and thus $\EE\lt[ |S_i \cap S_j| \rt] = \sum_{k=1}^{d} \EE\lt[ X_k \rt] = \frac{d^2}{|W|}$.
While the $X_k$ variables are not independent, they are negatively dependent, and thus we can use a Chernoff bound to show concentration.
For $\delta=\frac{\epsilon|W|}{2d}-1$, we get
\begin{align}
\Pr\lt[ |S_i \cap S_j| \ge (1+\delta) \EE\lt[ |S_i \cap S_j| \rt] \rt] 
&\le \lt( \frac{e^{\delta}}{(1+\delta)^{(1+\delta)}} \rt)^{\EE\lt[ |S_i \cap S_j| \rt]} \notag\\ 
&\le \lt( \frac{e^{\frac{\epsilon|W|}{2d}}}{\lt(\frac{\epsilon|W|}{2d}\rt)^{\frac{\epsilon|W|}{2d}}} \rt)^{{d^2}/{|W|}}\notag\\ 
&= \lt( \frac{\epsilon |W|}{2e d}\rt)^{-\epsilon d/2}\notag 
\end{align}
It follows by a union bound over the ${N \choose 2} \le \lt( \frac{\epsilon|W|}{2e d} \rt)^{\epsilon d/2}-1$ possible pairs of subsets that a set family $\mathcal{S}$ with the required bounded intersection property exists with nonzero probability.
Note that
\begin{align}
|\mathcal{S}| = \exp\lt(\log N\rt) 
&\ge 
\exp \lt( \frac{1}{2} \lt( \log\lt( \lt( \frac{\epsilon|W|}{2e d}  \rt)^{\epsilon d/2}- 1 \rt) - 2 \rt) \rt) \notag\\ 
\ann{since $d \le \frac{\epsilon |W|}{2e^{2+4/\epsilon}}$}
&\ge
\exp \lt( \frac{1}{2} \lt( \log\lt( \lt( e^{1+4/\epsilon}  \rt)^{\epsilon d/2}- 1 \rt) - 2 \rt) \rt) \notag\\
&\ge
\exp \lt( \frac{1}{2} \lt( \log\lt( \lt( e^{4/\epsilon}  \rt)^{\epsilon d/2} \rt) - 2 \rt) \rt) \notag\\
&\ge
\exp \lt( \frac{1}{2} \lt( \log\lt(  e^{2d}   \rt) - 2 \rt) \rt) \notag\\
&\ge
e^{ d  - 1  }.\notag
\end{align}
\end{proof}

%% file: det_algo.tex
\section{Proof of Theorem~\ref{thm:algo}} \label{app:algo}

\begin{reptheorem}{thm:algo}
\thmAlgo
\end{reptheorem}

\begin{proof}
    Consider an input vector $X$ with $\supp(X) = \{i_1, i_2, \ldots, i_s\}$
    where $i_1 < i_2 < \ldots < i_s$. 
    We can obtain a binary string of length $s = |\supp(X)|$ from $X$ by simply removing all $\perp$ from the vector. 
    In other words, the binary string of $X$ is 
    $X_{i_1}X_{i_2}\ldots X_{i_s}$. 
    For the rest of the proof, we use $X$ to refer to the input vector or its binary string representation depending on the context.

    We define $\pi_{i_a}(X)$ to be the mapping of the bit string of $X$ to a $(s-1)$-length bit string that we get by removing $X_{i_a}$. 
    That is, 
    \begin{align*}
        \pi_{i_a}(X_{i_1}X_{i_2}\ldots X_{i_a} \ldots X_{i_s}) = X_{i_1}X_{i_2}\ldots X_{i_{a-1}} X_{i_{a+1}} \ldots X_{i_s}. 
    \end{align*}

    Our algorithm determines the message sent by a party given the input vector $X$ to be $\pi_{i_a}(X)$, whereby $i_a$ is chosen in a way that guarantees that Charlie can compute the output correctly. 
    Note that, given $X$ and $\pi_{i_a}(X)$, the only bit values of $X$ that Charlie can not deduce with certainty is the value at index $i_a$, i.e., $X_{i_a}$.  
    Now, consider a valid input pair $\supp(X) = \{ i_1, \ldots, i_s\}$ and $\supp(Y) = \{j_1, \ldots, j_s\}$, where $\sigma = \supp(X) \cap \supp(Y)$. 
    Upon receiving the messages $\pi_{i_a}(X)$ and $\pi_{j_b}(Y)$ sent by Alice and Bob respectively, for some values $i_a$ and $j_b$, Charlie can deduce $X_\sigma$ and $Y_\sigma$ only if at least one of $i_a$ or $j_b$ is not $\sigma$. 

    In the following, we describe a protocol on how to choose the value $i_a \in \supp(X)$ so that Charlie can compute correctly: 
    We construct a partition of $\binom{[m]}{s}$ into blocks $\mathcal{R}_1, \ldots, \mathcal{R}_m$ such that for each $i \in [m]$: 
    \begin{itemize}
        \item[(a)] For all $I \in \mathcal{R}_i$, we have $i \in I$. 
        \item[(b)] For all $I, J$ in $\mathcal{R}_i$, it holds that $|I \cap J| \ge 2$.\footnote{The reader may notice that the set families  $\mathcal{R}_1, \ldots, \mathcal{R}_m$ are defined analogously as in the lower bound proof in Section~\ref{sec:overlap_lb}.} 
    \end{itemize}

    To see why Properties~(a) and (b) are sufficient, consider
    $I = \supp(X)$ and $J = \supp(Y)$ such that $I \cap J=\{\sigma\}$. 
    Let $i$ and $j$ are such that $I \in R_i$ and $J\in R_j$. 
    The message sent by Alice and Bob given input $X$ and $Y$ is thus $\pi_{i}(X)$ and $\pi_{j}(Y)$, respectively. 
    By the properties of the partition, it follows that $i \ne j$. 
    Hence, at least one of them is not equal to $\sigma$, which allow Charlie to deduce the correct value of $X_\sigma$ and $Y_\sigma$. %
    
    To complete the proof, we give a concrete construction of the blocks $\mathcal{R}_1, \ldots, \mathcal{R}_m$: 
    We partition $[m]$ into $\lceil\frac{m}{3}\rceil$ intervals, each of length $3$, except possibly one interval, where the length is either $1$ or $2$.
    We define $\Phi(b)$ to be the successor of $b$ in the interval containing $b$, following the cyclic order.  
    For instance, if the interval is $[b, b+1, b+2]$, then 
    \[
    \Phi(b+i) = b+ ((i+1)\bmod{3}),
    \] 
    for $i= 0, 1, 2.$
    Since we are considering cyclic order, the intervals are referred to as cycles.
    Note that if a cycle has only one element, say $a$, we set $\Phi(a)$ to be an arbitrary element different from $a$.
    
    Let $\mathcal{I} = \binom{[m]}{s}$. 
    We define $\mathcal{R}_i$ to be the set that contains all $I$ in $\mathcal{I}$ such that both $i$ and $\Phi(i)$ are in $I$, i.e., 
    $\mathcal{R}_i = \{I \in \mathcal{I} | \{i, \Phi(i)\} \subseteq I \}$.
    To ensure $\mathcal{R}_i$ and $\mathcal{R}_j$ are disjoint, we remove any $I$ that occurs in $\mathcal{R}_i \cap \mathcal{R}_j$ from one of them. 
    It is straightforward to verify that the resulting $\mathcal{R}_1, \ldots, \mathcal{R}_m$ indeed satisfy Properties~(a) and~(b).  
    
    It remains to show that $\bigcup_{i=1}^m \mathcal{R}_i =  \mathcal{I}$.   
    Consider $I \in \mathcal{I}$. 
    Recall that there are $s$ elements in $I$ and $\lceil\frac{m}{3}\rceil$ cycles. 
    Given that $s > \lceil\frac{m}{3}\rceil$  from the assumption, there must exists distinct elements $a$ and $b$ in $I$ such that they are in the same cycle. Hence $I$ is in some $\mathcal{R}_i$.    
\end{proof}

%% file: refs.bib
@article{DBLP:journals/corr/abs-2510-16336,
  author       = {Pachara Sawettamalya and
                  Huacheng Yu},
  title        = {A (Very) Nearly Optimal Sketch for \emph{k}-Edge Connectivity Certificates},
  journal      = {CoRR},
  volume       = {abs/2510.16336},
  year         = {2025},
  url          = {https://doi.org/10.48550/arXiv.2510.16336},
  doi          = {10.48550/ARXIV.2510.16336},
  eprinttype    = {arXiv},
  eprint       = {2510.16336},
  bibsource    = {dblp computer science bibliography, https://dblp.org}
}

@inproceedings{assadi2022lower,
  author    = {Assadi, Sepehr},
  title     = {Lower Bounds for Distributed Sketching},
  booktitle = {11th Workshop on Advances in Distributed Graph Algorithms (ADGA)},
  year      = {2022},
  url       = {https://adga-workshop.org/2022/assadi.pdf}
}

@article{hilton1967some,
  title={Some intersection theorems for systems of finite sets},
  author={Hilton, Anthony JW and Milner, Eric C},
  journal={The Quarterly Journal of Mathematics},
  volume={18},
  number={1},
  pages={369--384},
  year={1967},
  publisher={Oxford University Press}
}

@article{matouvsek2001probabilistic,
  title={The probabilistic method},
  author={Matou{\v{s}}ek, Ji{\v{r}}{\'\i} and Vondr{\'a}k, Jan},
  journal={Lecture Notes, Department of Applied Mathematics, Charles University, Prague},
  year={2001}
}

@inproceedings{MT16-podc,
  title={Brief announcement: deterministic graph connectivity in the broadcast congested clique},
  author={Montealegre, Pedro and Todinca, Ioan},
  booktitle={Proceedings of the 2016 ACM Symposium on Principles of Distributed Computing},
  pages={245--247},
  year={2016}
}

@inproceedings{PP20-fsttcs,
  title={Connectivity Lower Bounds in Broadcast Congested Clique},
  author={Pai, Shreyas and Pemmaraju, Sriram V},
  booktitle={40th IARCS Annual Conference on Foundations of Software Technology and Theoretical Computer Science (FSTTCS 2020)},
  year={2020},
  organization={Schloss Dagstuhl-Leibniz-Zentrum f{\"u}r Informatik}
}

@inproceedings{Y-soda21,
  author    = {Huacheng Yu},
  title     = {Tight Distributed Sketching Lower Bound for Connectivity},
  booktitle = {Proceedings of the 2021 {ACM-SIAM} Symposium on Discrete Algorithms,
               {SODA} 2021, Virtual Conference, January 10 - 13, 2021},
  pages     = {1856--1873},
  year      = {2021},
  url       = {https://doi.org/10.1137/1.9781611976465.111},
  doi       = {10.1137/1.9781611976465.111},
  bibsource = {dblp computer science bibliography, https://dblp.org}
}

@inproceedings{NY19-soda,
  author    = {Jelani Nelson and
               Huacheng Yu},
  title     = {Optimal Lower Bounds for Distributed and Streaming Spanning Forest
               Computation},
  booktitle = {Proceedings of the Thirtieth Annual {ACM-SIAM} Symposium on Discrete
               Algorithms, {SODA} 2019, San Diego, California, USA, January 6-9,
               2019},
  pages     = {1844--1860},
  year      = {2019},
  url       = {https://doi.org/10.1137/1.9781611975482.111},
  doi       = {10.1137/1.9781611975482.111},
  bibsource = {dblp computer science bibliography, https://dblp.org}
}

@inproceedings{JN-disc17,
  title={Brief announcement: on connectivity in the broadcast congested clique},
  author={Jurdzinski, Tomasz and Nowicki, Krzysztof},
  booktitle={31st International Symposium on Distributed Computing (DISC 2017)},
  year={2017},
  organization={Schloss Dagstuhl-Leibniz-Zentrum fuer Informatik}
}

@inproceedings{AGM-soda12,
  title={Analyzing graph structure via linear measurements},
  author={Ahn, Kook Jin and Guha, Sudipto and McGregor, Andrew},
  booktitle={Proceedings of the twenty-third annual ACM-SIAM symposium on Discrete Algorithms},
  pages={459--467},
  year={2012},
  organization={SIAM}
}

@inproceedings{BMRT-sirocco14,
  author    = {Florent Becker and
               Pedro Montealegre and
               Ivan Rapaport and
               Ioan Todinca},
  title     = {The Simultaneous Number-in-Hand Communication Model for Networks:
               Private Coins, Public Coins and Determinism},
  booktitle = {Structural Information and Communication Complexity - 21st International
               Colloquium, {SIROCCO} 2014, Takayama, Japan, July 23-25, 2014. Proceedings},
  pages     = {83--95},
  year      = {2014},
  url       = {https://doi.org/10.1007/978-3-319-09620-9\_8},
  doi       = {10.1007/978-3-319-09620-9\_8},
  bibsource = {dblp computer science bibliography, https://dblp.org}
}

@inproceedings{AKO-podc20,
  author    = {Sepehr Assadi and
               Gillat Kol and
               Rotem Oshman},
  title     = {Lower Bounds for Distributed Sketching of Maximal Matchings and Maximal
               Independent Sets},
  booktitle = {{PODC} '20: {ACM} Symposium on Principles of Distributed Computing,
               Virtual Event, Italy, August 3-7, 2020},
  pages     = {79--88},
  year      = {2020},
  url       = {https://doi.org/10.1145/3382734.3405732},
  doi       = {10.1145/3382734.3405732},
  bibsource = {dblp computer science bibliography, https://dblp.org}
}

@inproceedings{HKTZZ-focs19,
  author    = {Jacob Holm and
               Valerie King and
               Mikkel Thorup and
               Or Zamir and
               Uri Zwick},
  title     = {Random k-out Subgraph Leaves only O(n/k) Inter-Component Edges},
  booktitle = {60th {IEEE} Annual Symposium on Foundations of Computer Science, {FOCS}
               2019, Baltimore, Maryland, USA, November 9-12, 2019},
  pages     = {896--909},
  year      = {2019},
  url       = {https://doi.org/10.1109/FOCS.2019.00058},
  doi       = {10.1109/FOCS.2019.00058},
  bibsource = {dblp computer science bibliography, https://dblp.org}
}

@book{KN-book97,
 author = {Eyal Kushilevitz and Noam Nisan},
 title = {Communication Complexity},
 year = {1997},
 publisher = {Cambridge University Press}
}

@inproceedings{kapralov2017optimal,
  title={Optimal lower bounds for universal relation, and for samplers and finding duplicates in streams},
  author={Kapralov, Michael and Nelson, Jelani and Pachocki, Jakub and Wang, Zhengyu and Woodruff, David P and Yahyazadeh, Mobin},
  booktitle={2017 IEEE 58th Annual Symposium on Foundations of Computer Science (FOCS)},
  pages={475--486},
  year={2017},
  organization={Ieee}
}

@inproceedings{robinson2023distributed,
  title={Distributed Sketching Lower Bounds for k-Edge Connected Spanning Subgraphs, {BFS} Trees, and {LCL} Problems},
  author={Robinson, Peter},
  booktitle={37th International Symposium on Distributed Computing (DISC 2023)},
  pages={32--1},
  year={2023},
  organization={Schloss Dagstuhl--Leibniz-Zentrum f{\"u}r Informatik}
}

@article{kapralov2017optimalArXiv,
  title={Optimal lower bounds for universal relation, and for samplers and finding duplicates in streams},
  author={Kapralov, Michael and Nelson, Jelani and Pachocki, Jakub and Wang, Zhengyu and Woodruff, David P and Yahyazadeh, Mobin},
  journal={arXiv preprint arXiv:1704.00633},
  year={2017}
}

@book{jukna2001extremal,
  title={Extremal combinatorics: with applications in computer science},
  author={Jukna, Stasys},
  volume={29},
  year={2001},
  publisher={Springer}
}

@inproceedings{becker2011adding,
  title={Adding a referee to an interconnection network: What can (not) be computed in one round},
  author={Becker, Florent and Matamala, Martin and Nisse, Nicolas and Rapaport, Ivan and Suchan, Karol and Todinca, Ioan},
  booktitle={2011 IEEE International Parallel \& Distributed Processing Symposium},
  pages={508--514},
  year={2011},
  organization={IEEE}
}

@article{mcgregor2014graph,
  title={Graph stream algorithms: a survey},
  author={McGregor, Andrew},
  journal={ACM SIGMOD Record},
  volume={43},
  number={1},
  pages={9--20},
  year={2014},
  publisher={ACM New York, NY, USA}
}

@article{karchmer1995super,
  title={Super-logarithmic depth lower bounds via the direct sum in communication complexity},
  author={Karchmer, Mauricio and Raz, Ran and Wigderson, Avi},
  journal={Computational Complexity},
  volume={5},
  pages={191--204},
  year={1995},
  publisher={Springer}
}

@inproceedings{jurdzinski2018communication,
  title={Communication complexity in vertex partition whiteboard model},
  author={Jurdzinski, Tomasz and Lorys, Krzysztof and Nowicki, Krzysztof},
  booktitle={International Colloquium on Structural Information and Communication Complexity},
  pages={264--279},
  year={2018},
  organization={Springer}
}

@article{pinsker1964information,
  title={Information and information stability of random variables and processes},
  author={Pinsker, Mark S},
  journal={Holden-Day},
  year={1964}
}

@inproceedings{ahn2012graph,
  title={Graph sketches: sparsification, spanners, and subgraphs},
  author={Ahn, Kook Jin and Guha, Sudipto and McGregor, Andrew},
  booktitle={Proceedings of the 31st ACM SIGMOD-SIGACT-SIGAI symposium on Principles of Database Systems},
  pages={5--14},
  year={2012}
}

@inproceedings{guha2015vertex,
  title={Vertex and hyperedge connectivity in dynamic graph streams},
  author={Guha, Sudipto and McGregor, Andrew and Tench, David},
  booktitle={Proceedings of the 34th ACM SIGMOD-SIGACT-SIGAI Symposium on Principles of Database Systems},
  pages={241--247},
  year={2015}
}

@inproceedings{ashvinkumar2023evaluating,
  title={Evaluating Stability in Massive Social Networks: Efficient Streaming Algorithms for Structural Balance},
  author={Ashvinkumar, Vikrant and Assadi, Sepehr and Deng, Chengyuan and Gao, Jie and Wang, Chen},
  booktitle={APPROX/RANDOM},
  year={2023}
}
